\documentclass[aps,pra,superscriptaddress,twocolumn,longbibliography,showpacs]{revtex4-2}

\usepackage{pdfpages}

\makeatletter
\AtBeginDocument{\let\LS@rot\@undefined}
\makeatother

\usepackage{soul}
\setstcolor{red}
\usepackage{comment}

\usepackage{lipsum}
\usepackage{caption}
\DeclareCaptionJustification{justified}{\leftskip=0pt \rightskip=0pt \parfillskip=0pt plus 1fil}

\usepackage{graphicx}% Include figure files
\usepackage{dcolumn}% Align table columns on decimal point
\usepackage{bm}% bold math
\usepackage{caption}
\usepackage{amsmath,amssymb,amsthm}
\usepackage{braket}
\usepackage{empheq}
\usepackage{subfig}
\usepackage{tikz}
\usepackage{url}
\usepackage{hyperref}
\usepackage{quantikz}
\usepackage{xcolor}
\usepackage{float}
\usepackage{mathrsfs}

\definecolor{blueviolet}{rgb}{0.2, 0.2, 0.6}
\definecolor{webgreen}{rgb}{0,.5,0}
\definecolor{webbrown}{rgb}{.6,0,0}

\hypersetup{
    breaklinks=true,
    colorlinks=true,       % false: boxed links; true: colored links
    linkcolor=blueviolet,          % color of internal links
    citecolor=webbrown,        % color of links to bibliography
    filecolor=magenta,      % color of file links
    urlcolor=webgreen,           % color of external links
    runcolor=cyan
}

\newtheorem{theorem}{Theorem}
\newtheorem*{theorem*}{Theorem}

\newtheorem{lemma}{Lemma}
\newtheorem{proposition}{Proposition}

\usepackage{etoolbox}
\makeatletter
\patchcmd{\@citex}{\unskip, et al.}{\unskip et al.}{}{}
\patchcmd{\bib@device}
  {\addcontentsline{toc}{section}{\protect\numberline{}\refname}}
  {}{}{}
\makeatother

\begin{document}

\title{Randomized product formulas beyond optimal deterministic scaling}

\author{Leeseok Kim}
\affiliation{Center for Quantum Information and Control, University of New Mexico, NM 87131, USA}

\author{Luis Pedro García-Pintos}
\affiliation{Quantum and Condensed Matter Physics Group (T-4), Theoretical Division, Los Alamos National Laboratory, NM 87545, USA}

\begin{abstract}
Product formulas, also known as Trotter formulas, are among the most widely used and practical methods for simulating quantum systems on quantum computers. Here we introduce two new classes of randomized product formulas for simulating Hamiltonians with separated energy scales, $H=A+\alpha B$, where $\alpha$ is small. In the standard access model, where one can implement exponentials of $A$ and $B$ separately, our randomized formulas achieve $\mathcal O(\alpha^2)$ error scaling at the cost of only doubling the gate depth of the corresponding deterministic formula. We further prove an $\Omega(\alpha)$ lower bound for deterministic product formulas. In a stronger access model, allowing exponentials of $A+\alpha B_\ell$ for $B = \sum_{\ell}B_\ell$, our randomized formula, based on Trotter Heuristic Resource Improved Formulas for Time-dynamics (THRIFT)~[J. L. Bosse et al., Nat. Commun. 16, 2673 (2025)], achieves $\mathcal O(\alpha^3)$ error scaling with only constant-factor expected gate overhead. We also establish an $\Omega(\alpha^2)$ lower bound for deterministic product formulas in this access model. Numerical simulations confirm gate-count reductions for simulating physically motivated systems.
\end{abstract}

\maketitle

\textbf{\emph{Introduction.---}} Feynman’s vision of simulating quantum dynamics using computers governed by quantum mechanics~\cite{feynman1982simulating} has become a central motivation for quantum computing. Although many Hamiltonian-simulation algorithms have since been developed~\cite{childs2010on,childs2012hamiltonian,berry2015hamiltonian,berry2015simulating,low2017optimal,low2019hamiltonian}, Suzuki--Trotter formulas~\cite{trotter1959on,suzuki1976generalized,suzuki1990fractal}, also known as \emph{product formulas}, remain among the most practical methods owing to their simplicity and their often better-than-worst-case performance in practice~\cite{childs2018toward,childs2021theory}. Consequently, they have already been used to simulate many-body dynamics on current quantum processors~\cite{lanyon2011universal,martinez2016realtime,arute2020observation,kim2023evidence,cochran2025visualizing}.

Beyond its applications to various quantum protocols~\cite{viola2005random,santos2006enhanced,boixo2009eigenpath,wallman2016noise,wan2022randomized,martyn2025halving,gosset2025multi,yi2026faster,kim2026randomized,gunther2026phase,harrow2026randomized}, clever use of classical randomness has been shown to substantially improve the performance of product formulas~\cite{campbell2019random,childs2019faster,ouyang2020compilation,faehrmann2022randomizing,cho2024doubling,chen2025randomized}. The key mechanism is that averaging over multiple circuit realizations can cancel error terms that would persist if one repeatedly used a single  realization. This raises the following question:
\begin{center}
\emph{Can randomized product formulas surpass the intrinsic barriers of deterministic ones?}
\end{center}

In this Letter, we answer this question affirmatively for Hamiltonians with separated energy scales. We propose two new classes of randomized product formulas that {provably surpass deterministic barriers}. Concretely, we consider simulating Hamiltonians of the form
\begin{align}
H = A + \alpha B,
\end{align}
where $A$ and $B$ have comparable norms and $\alpha \ll1$, so that $\alpha B$ is a weak perturbation relative to $A$. Such Hamiltonians arise in many settings, including perturbation theory, open quantum systems, and physical systems with weak many-body interactions alongside strong one-body terms. Accordingly, efficient algorithms for simulating them have recently attracted interest~\cite{low2018hamiltonian,berry2020time,an2022timedependent,sharma2024hamiltonian,bosse2025efficient,bagherimehrab2026faster}.

Table~\ref{tab:summary} summarizes our main results for simulating $H$ up to time $t$ with $r$ Trotter steps. We study two access models. \textbf{(1)} In the standard access model, as in conventional product formulas, the Hamiltonians $A$ and $B$ are assumed to be ``easy'', so that their exponentials can be implemented directly. In this setting, standard $2k$-th order product formulas incur an $\mathcal{O}(\alpha t^{2k+1}/r^{2k})$ error: their leading error is governed by nested commutators of length $2k+1$ built from $A$ and $\alpha B$~\cite{childs2021theory}, and every nonvanishing such commutator contains at least one factor of $\alpha B$. By contrast, our randomized product formula achieves an $\mathcal{O}(\alpha^2 t^{2k+1}/r^{2k})$ error while using only twice as many gates as the corresponding deterministic formula. We further prove a no-go result showing that no deterministic formula built from finitely many such exponentials can generally obtain this $\mathcal{O}(\alpha^2)$ scaling. \textbf{(2)} In a stronger access model, motivated by recent work~\cite{bosse2025efficient}, we assume access to exponentials of $A$ and $A+\alpha B_\ell$, where $B=\sum_\ell B_\ell$. While deterministic $2k$-th order THRIFT achieves an error $\mathcal{O}(\alpha^2t^{2k+1}/r^{2k})$~\cite{bosse2025efficient}, our randomized product formula achieves an $\mathcal{O}(\alpha^3t^{2k+1}/r^{2k})$ error. We also establish a no-go result that no deterministic formula comprising finitely many such exponentials can generally achieve $\mathcal O(\alpha^3)$ scaling. %Moreover, it only incurs a constant-factor gate overhead compared to THRIFT~\cite{bosse2025efficient}. 

\begin{table*}[t!]
\caption{Summary of deterministic and randomized product-formula scalings for simulating $A+\alpha B$ with $\alpha\ll1$ over time $t$ using $r$ Trotter steps. We consider two access models: the standard access model used in the usual product-formula setting, with exponentials of $A$ and $\alpha B$, and the stronger access model, with exponentials of $A$ and $A+\alpha B_\ell$ where $B=\sum_\ell B_\ell$, motivated by Ref.~\cite{bosse2025efficient}. The no-go lower bound column reports the unavoidable leading dependence on $\alpha$ for any finite deterministic formula. Gate overhead is measured as the ratio of the number of elementary exponentials, in the corresponding access model, used by the randomized formula relative to the deterministic counterpart. Our randomized product formulas go beyond the corresponding deterministic no-go lower bounds while incurring only constant gate overhead.}
\label{tab:summary}
\begin{center}
\renewcommand{\arraystretch}{1.35}
\resizebox{\textwidth}{!}{
\begin{tabular}{@{}ccccc@{}}
\hline\hline
Access model
&
Deterministic formula
&
No-go lower bound
&
Randomized formula
&
Gate overhead
\\
\hline
\begin{tabular}[c]{@{}c@{}}
Standard access\\
$e^{-i sA}$, $e^{-i s\alpha B}$
\end{tabular}
&
\begin{tabular}[c]{@{}c@{}}
Product formulas\\
$\mathcal O\left(\alpha t^{2k+1}/r^{2k}\right)$
\end{tabular}
&
\begin{tabular}[c]{@{}c@{}}
%No-go lower bound\\
(Proposition~\ref{prop:deterministic-barrier-standard})\\
$\Omega(\alpha)$
\end{tabular}
&
\begin{tabular}[c]{@{}c@{}}
Randomized product formula\\
(Theorem~\ref{thm:randomized-high-order})\\
$\mathcal O\left(\alpha^2 t^{2k+1}/r^{2k}\right)$
\end{tabular}
&
\begin{tabular}[c]{@{}c@{}}
$<2$
\end{tabular}
\\
\hline
\begin{tabular}[c]{@{}c@{}}
Stronger access\\
$e^{-i sA}, e^{-i s(A+\alpha B_\ell)}$
\end{tabular}
&
\begin{tabular}[c]{@{}c@{}}
THRIFT~\cite{bosse2025efficient}\\
$\mathcal O\left(\alpha^2t^{2k+1}/r^{2k}\right)$
\end{tabular}
&
\begin{tabular}[c]{@{}c@{}}
%No-go lower bound\\
(Proposition~\ref{prop:strong-access-barrier}) \\
$\Omega(\alpha^2)$
\end{tabular}
&
\begin{tabular}[c]{@{}c@{}}
Randomized THRIFT\\
(Theorem~\ref{thm:randomized-thrift})\\
$\mathcal O \left(\alpha^3t^{2k+1}/r^{2k}\right)$
\end{tabular}
&
\begin{tabular}[c]{@{}c@{}}
$\mathcal{O}(1)$
\end{tabular}
\\
\hline\hline
\end{tabular}%
}
\end{center}
\end{table*}

\textbf{\emph{Standard access model.---}} We first consider the usual product-formula access model: for any real $s$, we can implement $e^{-isA}$ and $e^{-is\alpha B}$.

\textbf{\emph{(i) Randomized first-order formula:}} For any real $\tau$ and $u\in[0,1]$, define the shifted first-order formula
\begin{align}
\label{eq:first-order-randomized-trotter-single-step}
\mathscr R_1(\tau,u) := e^{-i(1-u)\tau A}e^{-i\alpha\tau  B}e^{-iu\tau A}.
\end{align}
At $u=0$ and $u=1$, $\mathscr R_1$ recovers the two standard first-order product formulas, so $u$ continuously shifts the $B$-step between them.

\begin{figure}[b]
\centering
\includegraphics[width=8.6cm]{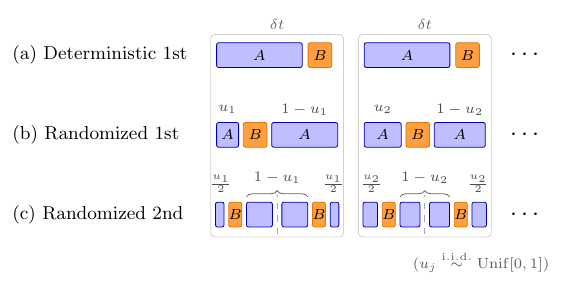}
\caption{Schematics of deterministic and randomized product formulas for simulating $A+\alpha B$ over repeated time steps $\delta t$ in the standard access model.  (a) The deterministic 1st-order formula, $(e^{-i\delta t \alpha B} e^{-i\delta t A})^r$. (b) One instance of the 1st-order randomized formula, Eq.~\eqref{eq:randomized-1st-order}. (c) One instance of the 2nd-order randomized formula, Eq.~\eqref{eq:randomized-2nd-order}, obtained by symmetrizing (b); higher-order formulas are recursively constructed from this via Eq.~\eqref{eq:randomized-high-order}.}
\label{fig:randomized-product-formulas}
\end{figure}

The randomized formula is obtained by sampling this shift independently at each step. Let $\delta t=t/r$ be a single Trotter step and $\mathbf u=(u_1,\ldots,u_r)$ with $u_j\sim\operatorname{Unif}[0,1]$ independently. Define 
\begin{align}
\label{eq:randomized-1st-order}
\mathscr R_1^{(r)}(t,\mathbf u) := \mathscr R_1(\delta t,u_r)\cdots \mathscr R_1(\delta t,u_1),
\end{align}
as shown schematically in Fig.~\ref{fig:randomized-product-formulas}. After averaging over the sampled shifts, the resulting evolution is described by the mixed-unitary channel
\begin{align}
\mathcal R_1^{(r)}(\rho) := \mathbb E_{\mathbf u}\left[\mathscr R_1^{(r)}(t,\mathbf u)\; \rho \; \mathscr R_1^{(r)}(t,\mathbf u)^\dagger\right].
\end{align}
Surprisingly, this simple randomization yields 
\begin{align}
\label{eq:randomized-1st-order-bound}
\left\|\mathcal R_1^{(r)}-\mathcal U_H(t) \right\|_\diamond = \mathcal O \left(\frac{\alpha^2t^3}{r^2}\right),
\end{align}
where $\mathcal U_H(t)(\rho) := e^{-itH}\rho e^{itH}$, which demonstrates the scaling improvement in \emph{both} $\alpha$ and $t$. (See the Supplemental Material (SM) for the derivation.) Here, $\|A\|:=\sup_{\|\ket{\psi}\|_2=1}\|A\ket{\psi}\|_2$ denotes the operator (spectral) norm, and $\|\Phi\|_\diamond := \sup_{\rho} \left\|(\Phi\otimes\mathcal I_d)(\rho)\right\|_1$ denotes the diamond norm~\cite{watrous2018theory}. By completely removing the $\mathcal O(\alpha)$ contribution, the randomized first-order formula has potential to outperform deterministic higher-order formulas that scale as $\mathcal{O}(\alpha t^{2k+1}/r^{2k})$, in sufficiently perturbative regimes. Moreover, the randomization comes at essentially \emph{no extra gate overhead}: after merging adjacent $A$-evolutions, each realization uses $r$ weak $B$-evolutions and $r+1$ $A$-evolutions, only one more $A$-block than the deterministic first-order formula.
%compared with $r$ weak $B$-evolutions and $r$ $A$-evolutions for the deterministic first-order product formula.

To see why the $\mathcal O(\alpha)$ term cancels, it suffices to analyze the error in a single Trotter step. (Independence of the samples then lets the estimate telescope over $r$ steps.) Let $B_I(s)=e^{isA}Be^{-isA}$. In the interaction picture with respect to $A$, the ideal evolution can be written as
\begin{align}
\label{eq:exact-evolution-alpha-exapnsion}
e^{-i\delta t H} = e^{-i\delta t A} \!\left(I \! - \!i\alpha\int_0^{\delta t} B_I(s) ds + \mathcal O(\alpha^2\delta t^2)\!\right),
\end{align}
while Eq.~\eqref{eq:first-order-randomized-trotter-single-step} satisfies
\begin{align}
\label{eq:first-order-randomized-expansion}
\mathscr R_1(\delta t,u) &= e^{-i\delta t A} \left(e^{iu\delta t A}e^{-i\alpha\delta t B}e^{-iu\delta t A}\right)  \nonumber \\
&= e^{-i\delta t A} \left(I-i\alpha\delta t B_I(u\delta t) + \mathcal O(\alpha^2\delta t^2)\right).
\end{align}
Averaging over $u\sim\operatorname{Unif}[0,1]$ and using $s=u\delta t$ gives
\begin{align}
\label{eq:first-order-randomized-expansion-match}
\delta t \mathbb E_u[B_I(u\delta t)] = \int_0^{\delta t} B_I(s) ds,
\end{align}
so Eq.~\eqref{eq:first-order-randomized-trotter-single-step} agrees with the ideal evolution through first order in $\alpha$ on average.

Finally, the cancellation of the $\mathcal O(\delta t^2)$ term can be viewed as a continuous analogue of the randomized forward/reverse ordering of Ref.~\cite{childs2019faster}, with the discrete ordering replaced by a uniform insertion point $u\in[0,1]$.

\textbf{\emph{(ii) Randomized high-order formula:}} Define the shifted second-order formula by symmetrizing Eq.~\eqref{eq:first-order-randomized-trotter-single-step}
\begin{align}
\label{eq:randomized-2nd-order}
\mathscr R_2(\tau,u) &:= \mathscr R_1(\tau/2,1-u)\mathscr R_1(\tau/2,u), %\nonumber \\
%&= e^{-i(1-u)\tau \frac{A}{2}}e^{-i\alpha\tau \frac{B}{2}} e^{-iu\tau A} e^{-i\alpha\tau \frac{B}{2}}e^{-i(1-u)\tau \frac{A}{2}}.
\end{align}
as illustrated in Fig.~\ref{fig:randomized-product-formulas}(c). Following Suzuki's recursive construction~\cite{suzuki1990fractal}, define
\begin{align}
\label{eq:randomized-high-order}
\mathscr R_{2k}(\tau,u) :={} & \mathscr R_{2k-2}(s_k\tau,u)^2 \mathscr R_{2k-2}((1-4s_k)\tau,u) \nonumber\\
&\times \mathscr R_{2k-2}(s_k\tau,u)^2 
\end{align}
for $k \geq 2$, where $s_k=1/(4-4^{1/(2k-1)})$. For simulating time $t$ with $\delta t = t/r$ and $\mathbf u=(u_1,\ldots,u_r)$, define
\begin{align}
\mathscr R_{2k}^{(r)}(t,\mathbf u) := \mathscr R_{2k}(\delta t,u_r)\cdots
\mathscr R_{2k}(\delta t,u_1).
\end{align}
Averaging over $u_j \sim \operatorname{Unif}[0,1]$ gives the mixed-unitary channel
\begin{align}
\label{eq:randomized-high-order-formula}
\mathcal R_{2k}^{(r)}(\rho) := \mathbb E_{\mathbf u}\left[\mathscr R_{2k}^{(r)}(t,\mathbf u) \; \rho \; \mathscr R_{2k}^{(r)}(t,\mathbf u)^\dagger \right].
\end{align}
This randomized protocol retains the $2k$-th order scaling in $\delta t$, while improving the dependence on $\alpha$ to $\mathcal O(\alpha^2)$.
\begin{theorem}
\label{thm:randomized-high-order}
For each fixed $k\ge1$, the randomized $2k$-th order formula given in Eq.~\eqref{eq:randomized-high-order-formula} yields
\begin{align}
\left\|\mathcal R_{2k}^{(r)}-\mathcal U_H(t) \right\|_\diamond = \mathcal O \left( \frac{\alpha^2 t^{2k+1}}{r^{2k}} \right).
\end{align}
\end{theorem}
The proof, given in the Supplemental Material (SM), has two ingredients: the usual Suzuki recursion recovers the desired order $2k$, while the random shift cancels the term linear in $\alpha$ by the same interaction-picture mechanism as above. Crucially, this requires only \emph{twice the gate depth}: one $\mathscr R_{2k}$ step uses $2\times 5^{k-1}$ weak $B$-evolutions, compared with $5^{k-1}$ for the deterministic Suzuki formula of the same order, up to merging adjacent $A$-evolutions. Thus, the protocol merely doubles the circuit depth while achieving a quadratic improvement in the $\alpha$-scaling, substantially reducing the error for small $\alpha$.

\textbf{\emph{(iii) Deterministic barrier:}} We now show that the improvement above fundamentally relies on randomization. In particular, under the same standard access model, no finite deterministic product formula can universally eliminate the $\mathcal O(\alpha)$ error. 

\begin{proposition}
\label{prop:deterministic-barrier-standard}
Fix $t\neq 0$ and a finite $m\in\mathbb N$. For real coefficients $a_1,\ldots,a_{m+1}$ and $b_1,\ldots,b_m$, define
$\mathscr D_m(t,\alpha) := e^{-ia_{m+1}tA} e^{-ib_m\alpha tB} \cdots e^{-ib_1\alpha tB} e^{-ia_1tA}$. Let $\mathcal D_m(t,\alpha)(\rho) := \mathscr D_m(t,\alpha)\rho \mathscr D_m(t,\alpha)^\dagger$. For every choice of coefficients independent of $A$, $B$, and $\alpha$, there exist Hermitian matrices $A$ and $B$ such that
\begin{align}
\left\| \mathcal D_m(t,\alpha) - \mathcal U_H(t) \right\|_\diamond = \Omega(\alpha).
\end{align}
\end{proposition}
The proof is given in the SM. The intuition is simple: a deterministic formula samples the interaction-picture perturbation at only finitely many fixed times, and hence cannot reproduce the continuum average $\int_0^t B_I(s)ds$ for all $A$ and $B$. Randomization evades this obstruction by sampling the same average unbiasedly.

%\textbf{\emph{(iv) Robustness to noise:}}

\textbf{\emph{Stronger access model.---}} We next consider a stronger access model, motivated by Ref.~\cite{bosse2025efficient}, in which
\begin{align}
B=\sum_{\ell=1}^L B_\ell
\end{align}
and, for any real $s$, we can implement both $e^{-i sA}$ and $e^{-i s(A+\alpha B_\ell)}$ for each $\ell$. 

\textbf{\emph{(i) THRIFT:}} Under this access model, THRIFT~\cite{bosse2025efficient} implements a product formula directly in the interaction picture with respect to $A$. A single THRIFT step is
\begin{align}
\mathscr T(\tau) := e^{-i \tau A}\prod_{\ell=1}^{L} \left(e^{i \tau A}e^{-i \tau(A+\alpha B_{\ell})}\right),
\label{eq:thrift}
\end{align}
where $\prod_{\ell=1}^{L} X_\ell := X_L X_{L-1}\cdots X_1$.

The key point is that each factor inside the product is an interaction-picture evolution. Define $B_\ell^{(A)}(s):=e^{i sA}B_\ell e^{-i sA}$, and, for any interval $I\subset[0,\tau]$,
\begin{align}
W_\ell(I) := \mathcal T\exp\left(-i\alpha\int_I B_\ell^{(A)}(s) ds \right).
\end{align}
For an interval $[x,y]\subset[0,\tau]$, this is implementable as
\begin{align}
\label{eq:W-interval-implementation}
W_\ell([x,y]) = e^{iyA}e^{-i(y-x)(A+\alpha B_\ell)}e^{-ixA},
\end{align}
using only the gates allowed in the stronger access model. In particular, $e^{i\tau A}e^{-i\tau(A+\alpha B_\ell)}=W_\ell([0,\tau])$, and thus the THRIFT step in the interaction picture is
\begin{align}
S(\tau) := e^{i\tau A}\mathscr T(\tau) = \prod_{\ell=1}^{L} W_{\ell}([0,\tau]),
\end{align}
which matches the $\mathcal O(\alpha)$ term of the ideal interaction-picture evolution. Defining $\mathscr T^{(r)}(t):=\mathscr T(\delta t)^r$ with $\delta t=t/r$, one therefore obtains
\begin{align}
\|\mathscr T^{(r)}(t)-e^{-i tH}\| = \mathcal O\left(\frac{\alpha^2t^2}{r}\right).
\end{align}
Ref.~\cite{bosse2025efficient} extended this to achieve higher-order scaling in $t$ by constructing a second-order formula by symmetrizing Eq.~\eqref{eq:thrift} and applying Suzuki recursion.

\begin{figure}[t]
\centering
\includegraphics[width=0.82\columnwidth]{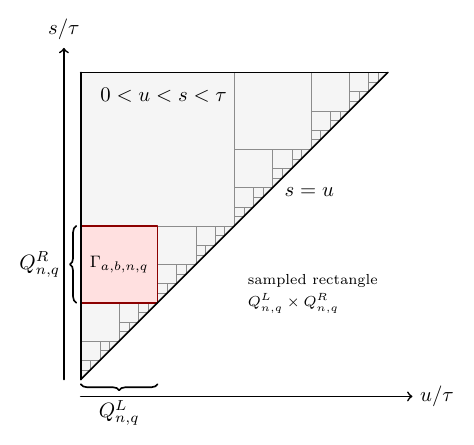}
\caption{Dyadic sampling for randomized THRIFT. Sampled rectangles in the time-ordering triangle specify the local corrections in Eq.~\eqref{eq:local-group-commutator}.}
\label{fig:dyadic-thrift-sampling}
\end{figure}

\textbf{\emph{(ii) Randomized THRIFT:}} The leading error of a THRIFT step, relative to the ideal interaction-picture evolution, appears at order $\alpha^2$. Its coefficient is
\begin{align}
\label{eq:thrift-second-order-error}
\sum_{1\le a<b\le L} \int_{0 < u < s < \tau}  \left[B_{a}^{(A)}(s), B_{b}^{(A)}(u) \right] duds.
\end{align}
(See the SM for derivation.) We cancel this term on average by sampling local correction unitaries over the triangle $0<u<s<\tau$.

Decompose this triangle into dyadic rectangles $Q_{n,q}^L\times Q_{n,q}^R$, as illustrated in Fig.~\ref{fig:dyadic-thrift-sampling}, where $n$ denotes the dyadic level and $q = 0, \dotsc, 2^{n-1}$ indexes the rectangles within that level. For each $1\le a<b\le L$, and $m\in\mathbb N$, define the local correction unitary
\begin{align}
\label{eq:local-group-commutator}
C_{a,b,n,q}^{[m]} :={}& W_{a}(Q_{n,q}^R)^mW_{b}(Q_{n,q}^L)^m \nonumber  \\ 
& \qquad \times W_{a}(Q_{n,q}^R)^{-m} W_{b}(Q_{n,q}^L)^{-m}.
\end{align}
By Eq.~\eqref{eq:W-interval-implementation}, each factor in Eq.~\eqref{eq:local-group-commutator} is implementable using only the gates allowed in the stronger access model.

By the group-commutator identity~\cite{kitaev2002classical}, 
\begin{align}
\label{eq:local-group-commutator-expansion}
&C_{a,b,n,q}^{[m]} = I-\alpha^2m^2\Gamma_{a,b,n,q} + \mathcal O(\alpha^3), \\
&\Gamma_{a,b,n,q} = \int_{s\in Q_{n,q}^R}
\int_{u\in Q_{n,q}^L} \left[B_{a}^{(A)}(s), B_{b}^{(A)}(u) \right]duds. \nonumber 
\end{align}
Thus $C_{a,b,n,q}^{[m]}$ produces the negative of the local THRIFT commutator error, amplified by $m^2$.

Choose $m_n=\lceil c2^{\beta n}\rceil$, where the order-dependent exponent $\beta$ is chosen to guarantee the desired high-order error scaling and finite expected gate cost, while $c>0$ is chosen so that the probabilities below form a valid distribution. We apply the correction indexed by $(a,b,n,q)$ with probability $p_{a,b,n,q}:=1/{m_n^2}$ and otherwise apply no correction. %The normalization of this distribution and the finite expected gate count are analyzed in the SM. 
The resulting interaction-picture step is
\begin{align}
\label{eq:randomized-thrift-step-ip}
S_{\rm rand}(\tau,\omega) :=
\begin{cases}
S(\tau), & \omega=0,\\[1mm]
C_{a,b,n,q}^{[m_n]}S(\tau), & \omega=(a,b,n,q).
\end{cases}
\end{align}
%A sampled correction $C_{a,b,n,q}^{[m_n]}$ contributes $4m_n$ additional $W$-blocks, i.e., an extra gate depth $\mathcal{O}(m_n)$, while the $\omega=0$ branch has the same depth as the original THRIFT step. 
Although the sampling distribution has infinite support over $n$, $\Pr[n=\infty]=0$, so each sampled step has finite gate depth with probability $1$. Since the dyadic rectangles partition the time-ordering triangle $0<u<s<\tau$, Eq.~\eqref{eq:local-group-commutator-expansion} implies that the averaged second-order contribution of the sampled correction is exactly the negative of the $\mathcal O(\alpha^2)$ term in Eq.~\eqref{eq:thrift-second-order-error}. 

Returning to the Schr\"odinger picture,
\begin{align}
\label{eq:first-order-randomized-thrift}
\mathscr T_{\rm rand}(\tau,\omega) := e^{-i\tau A}S_{\rm rand}(\tau,\omega).
\end{align}
For $\delta t=t/r$ and independent samples $\boldsymbol\omega=(\omega_1,\ldots,\omega_r)$, define
\begin{align}
\mathscr T_{\rm rand}^{(r)}(t,\boldsymbol\omega) := \mathscr T_{\rm rand}(\delta t,\omega_r)\cdots \mathscr T_{\rm rand}(\delta t,\omega_1),
\end{align}
and the corresponding mixed-unitary channel
\begin{align}
\label{eq:randomized-thrift-formula}
\mathcal T_{\rm rand}^{(r)}(\rho) := \mathbb E_{\boldsymbol\omega} \left[\mathscr T_{\rm rand}^{(r)}(t,\boldsymbol\omega) \; \rho \; \mathscr T_{\rm rand}^{(r)}(t,\boldsymbol\omega)^\dagger \right].
\end{align}
Then, one has
\begin{align}
\label{eq:randomized-thrift-bound}
\|\mathcal T_{\rm rand}^{(r)}-\mathcal U_H(t)\|_\diamond = \mathcal O\left(\frac{\alpha^3t^3}{r^2}\right).
\end{align}

\textbf{\emph{(iii) Randomized high-order THRIFT:}} To extend the time-order scaling, we first construct a second-order formula by symmetrizing Eq.~\eqref{eq:first-order-randomized-thrift}
\begin{align}
\label{eq:randomized-thrift-second-order} \mathscr T_{{\rm rand},2}(\tau,\omega) := \mathscr T_{\rm rand}(\tau/2,\omega) \mathscr T_{\rm rand}(-\tau/2,\omega)^\dagger.
\end{align}
For $k\geq2$, Suzuki recursion gives
\begin{align}
\label{eq:randomized-high-order-thrift}
\mathscr T_{{\rm rand},2k}(\tau,\omega) ={}& \mathscr T_{{\rm rand},2k-2}(s_k\tau,\omega)^2 \mathscr T_{{\rm rand},2k-2}((1-4s_k)\tau,\omega) \nonumber\\ &\times \mathscr T_{{\rm rand},2k-2}(s_k\tau,\omega)^2,
\end{align}
where $s_k=1/(4-4^{1/(2k-1)})$. The same sample $\omega$ is reused in a single Trotter step. For $\delta t=t/r$, let 
\begin{align}
\label{eq:randomized-high-order-thrift-channel}
\mathcal T_{{\rm rand},2k}^{(r)}(\rho) := \mathbb E_{\boldsymbol\omega} \left[\mathscr T_{{\rm rand}, 2k}^{(r)}(t,\boldsymbol\omega)  \rho  \mathscr T_{{\rm rand}, 2k}^{(r)}(t,\boldsymbol\omega)^\dagger \right].
\end{align}
denote the corresponding averaged $r$-step channel. 

\begin{theorem}%[Randomized THRIFT]
\label{thm:randomized-thrift}
For each fixed $k\geq1$, the randomized $2k$-th order THRIFT formula defined in Eq.~\eqref{eq:randomized-high-order-thrift-channel} yields
\begin{align}
\left\| \mathcal T_{{\rm rand},2k}^{(r)} -\mathcal U_H(t) \right\|_\diamond = \mathcal O\left( \frac{\alpha^3t^{2k+1}}{r^{2k}} \right).
\end{align}
Moreover, the average number of gates per step is $\mathcal O(L)$.
\end{theorem}
%The proof, including the exact probability distribution and the gate-depth analysis, is given in the SM. Randomization therefore cancels the deterministic $\mathcal O(\alpha^2)$ commutator error on average, achieving the $\mathcal O(\alpha^3)$ error. Furthermore, the randomized THRIFT formula uses an average gate count of $\mathcal{O}(L)$ per step, increasing the gate depth by only a constant factor compared with the original THRIFT formula that uses $2L-1$ gates per step.
The proof, including the precise order-dependent sampling distribution and the expected gate-depth analysis, is given in the SM. In particular, for each fixed order, randomized THRIFT uses $\mathcal O(L)$ gates per step on average, incurring only a constant-factor overhead relative to its deterministic counterpart.

\textbf{\emph{(iv) Deterministic barrier:}} We now show that no finite product of exponentials allowed by the stronger access model can universally achieve an $\mathcal{O}(\alpha^3)$ error bound.

\begin{proposition}
\label{prop:strong-access-barrier}
Fix $t\ne0$, $L\ge2$, and a finite $M\in\mathbb N$.  Set $B_0:=0$.  For real
coefficients $c_1,\ldots,c_M$ and labels $\ell_1,\ldots,\ell_M\in\{0,1,\ldots,L\}$, define $\mathscr D_M(t,\alpha):= e^{-ic_Mt(A+\alpha B_{\ell_M})}\cdots e^{-ic_1t(A+\alpha B_{\ell_1})}$. Let $\mathcal D_M(t,\alpha)(\rho):= \mathscr D_M(t,\alpha) \rho \mathscr D_M(t,\alpha)^\dagger$. For every choice of coefficients and labels independent of $A,B_1,\ldots,B_L$, and $\alpha$, there exist Hermitian matrices $A,B_1,\ldots,B_L$ such that
\begin{align}
\left\|\mathcal D_M(t,\alpha) - \mathcal{U}_H(t)\right\|_\diamond=\Omega(\alpha^2).
\label{eq:main-strong-no-go-error}
\end{align}
\end{proposition}

The proof is given in the SM. This result extends the no-go result proven in Ref.~\cite[Theorems~10 and 11]{bosse2025efficient} to arbitrary finite circuits in the stronger access model.

\begin{figure}[t!]
    \centering
    \includegraphics[width=0.82\columnwidth]{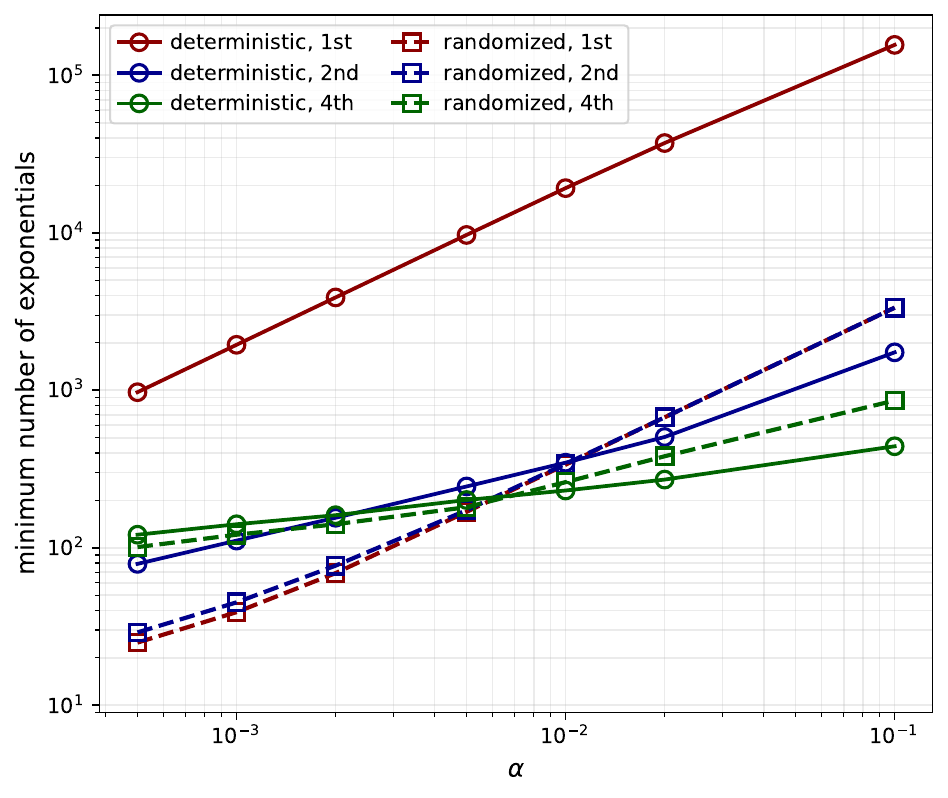}
    \caption{Minimum number of elementary exponentials (gates) required to simulate the dynamics of the 9-qubit transverse-field Ising model at $t=9/E_0$, comparing standard deterministic product formulas with our randomized formulas as a function of $\alpha$. Our approach requires fewer gates for small $\alpha$, while deterministic formulas perform better as $\alpha$ increases since our method uses twice as many gates.}
    \label{fig:tfim-numerics}
\end{figure}

\textbf{\emph{Numerical simulations.---}} We now present numerical experiments showing that our results are not only theoretically interesting but also practically useful. In all simulations, averages over the randomized formulas are computed deterministically (see the SM for more details and additional experiments). We first benchmark our randomized product formulas in the standard-access setting by simulating the dynamics of the $9$-qubit 2D nearest-neighbor transverse-field Ising model on an open $3\times3$ square lattice, $H=A+\alpha B$ with $A=E_0\sum_{j=1}^{9} Z_j, B=E_0\sum_{\langle j,k\rangle}X_jX_k$, for time $t = 9/E_0$ where $E_0>0$ sets the energy scale. We set $E_0=1$ in the numerics. For each $\alpha$, we find the minimum number of exponentials of $A$ and $\alpha B$ required for the upper bound on the diamond error to be at most $\epsilon = 10^{-4}$. The results are shown in Fig.~\ref{fig:tfim-numerics}. For a $2k$th-order formula and fixed target error $\epsilon$, the number of Trotter steps scales as $\mathcal O((\alpha t^{2k+1}/\epsilon)^{1/2k})$ and $\mathcal  O((\alpha^2t^{2k+1}/\epsilon)^{1/2k})$. Thus, for fixed $t$ and $\epsilon$, the gate count scales as $\mathcal O(\alpha^{1/2k})$ and $\mathcal O(\alpha^{1/k})$, explaining why randomization is most advantageous at small $\alpha$. In this regime, the randomized first-order formula can even outperform deterministic higher-order formulas. As $\alpha$ increases, this advantage diminishes, and the twofold gate overhead of the higher-order randomized formulas can make them more costly.

%For small $\alpha$, the randomized formulas require fewer gates than their deterministic counterparts, as expected. In fact, the randomized first-order formula can outperform deterministic higher-order formulas. As $\alpha$ increases, this advantage diminishes, and the twofold gate overhead can make the randomized formulas more costly.

\begin{figure}[t!]
    \centering
    \includegraphics[width=0.86\linewidth]{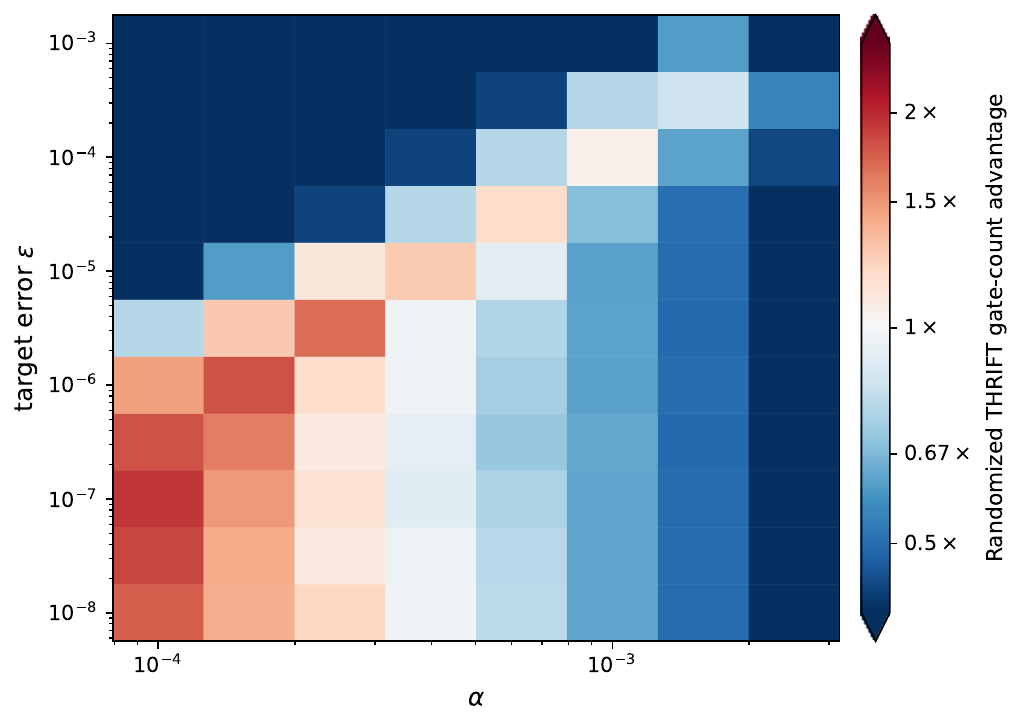}
    \caption{Gate-count advantage of randomized first-order THRIFT for simulating the dynamics of the $8$-qubit random-field Heisenberg chain at $t=9/E_0$, across interaction strengths $\alpha$ and target errors $\epsilon$. The color shows the ratio of the smaller median gate count of deterministic first- and second-order THRIFT to the median expected gate count of randomized first-order THRIFT. For each method, the median is taken over ten independent field realizations of the minimum gate count required to achieve the target error. In the red region, randomized first-order THRIFT requires fewer gates than either first- or second-order deterministic THRIFT.}
    \label{fig:randomized_thrift_t9_phase_map}
\end{figure}

We next benchmark randomized THRIFT in the stronger-access setting using an open $8$-qubit random-field Heisenberg chain. Defining $K_j=X_jX_{j+1}+Y_jY_{j+1}+Z_jZ_{j+1}$, we consider $H=A+\alpha(B_0+B_1)$, where $A=E_0\sum_{j=1}^{8}h_jZ_j$, and the interaction terms are partitioned into two alternating bond layers, $B_0=E_0(K_1+K_3+K_5+K_7)$, and $B_1=E_0(K_2+K_4+K_6)$. As before, we set $E_0 = 1$ in the numerics. The coefficients $h_j$ are sampled independently and uniformly from $[-1,1]$. We assume access to exponentials of $A$, $A+\alpha B_0$, and $A+\alpha B_1$. For each $(\alpha,\epsilon)$ and ten independent field realizations, we determine the minimum number of exponentials required for the certified diamond-error bound to be at most $\epsilon$ at $t=9/E_0$. We compare deterministic first- and second-order THRIFT with randomized first-order THRIFT defined in Eq.~\eqref{eq:randomized-thrift-formula}. Figure~\ref{fig:randomized_thrift_t9_phase_map} reveals a weak-coupling, high-precision regime in which randomized first-order THRIFT is cheaper than both deterministic formulas. At fixed $t$, its $\mathcal{O}(\alpha^{3/2}/\epsilon^{1/2})$ cost scales more favorably with $\epsilon$ than first-order THRIFT and with $\alpha$ than second-order THRIFT in the weak-coupling limit, provided that these scaling advantages outweigh the randomized overhead. We expect the observed location of this crossover at relatively small $\alpha$ to reflect a sizable prefactor in the $\mathcal{O}(\alpha^3)$ error.

\textbf{\emph{Conclusion.---}} We introduce two novel classes of randomized product formulas for simulating Hamiltonians of the form $H=A+\alpha B$, where $\alpha$ is small. For each access model, we establish fundamental limitations of any deterministic formulas through no-go results and show that our randomized formulas overcome them. Moreover, our formulas incur only a constant-factor gate overhead relative to their deterministic counterparts. Numerical simulations of physically motivated systems confirm their gate-count advantages, demonstrating both their theoretical significance and potential practical utility.

As in many applications of randomization~\cite{santos2006enhanced,wallman2016noise,yi2026faster}, our randomized product formulas may be intrinsically robust to certain types of noise, which would further strengthen their practical motivation. It would also be valuable to benchmark both formulas across other classes of Hamiltonians.
In addition, although known quadratic lower bounds~\cite{akibue2024probabilistic,braasch2026limits} suggest the optimality of our standard-access construction in its $\alpha$-scaling, the optimality of randomized THRIFT remains open. More broadly, it is worth exploring whether randomization can provably overcome the fundamental performance limitations of deterministic protocols in other settings.

\textbf{\emph{Acknowledgements.---}} L.K. acknowledges assistance from ChatGPT (GPT-5.5 and GPT-5.6 Sol), primarily in refining figures and assisting with the implementation of numerical simulations. L.K. acknowledges support by the U.S. Department of Energy, Office of Science, Office of Advanced Scientific Computing Research under Contract No. DE-AC05-00OR22725 through the Accelerated Research in Quantum Computing Program MACH-Q project. L.P.G.P. acknowledges support by the U.S. Department of Energy, Office of Science, Basic Energy Sciences program (award No. DE-SCL0000157).

%\bibliography{ref}

\addtocontents{toc}{\protect\setcounter{tocdepth}{-1}}
\bibliography{ref}
\addtocontents{toc}{\protect\setcounter{tocdepth}{2}}

\clearpage

%%%%%%%%%%%%%%%%%%%%%%%%%%%%%%%%%%%%%%%%%%%%%%%%%%%%%%%%%%%%%%%%%%%%%%%%%%%%%%
% -------------  SUPPLEMENTAL MATERIAL  -------------
%%%%%%%%%%%%%%%%%%%%%%%%%%%%%%%%%%%%%%%%%%%%%%%%%%%%%%%%%%%%%%%%%%%%%%%%%%%%%%%

\clearpage
\onecolumngrid
\pagenumbering{arabic}

\begin{center}
\textbf{\large Supplemental Material}
\end{center}

\setcounter{section}{0}
\renewcommand{\thesection}{S\arabic{section}}

\setcounter{subsection}{0}
\renewcommand{\thesubsection}{\thesection.\arabic{subsection}}

\setcounter{table}{0}
\renewcommand{\thetable}{S\arabic{table}}

\setcounter{figure}{0}
\renewcommand{\thefigure}{S\arabic{figure}}

\setcounter{equation}{0}
\renewcommand{\theequation}{S\arabic{equation}}

\setcounter{secnumdepth}{3}

\setcounter{theorem}{0}
\renewcommand{\thetheorem}{S\arabic{theorem}}

\setcounter{proposition}{0}
\renewcommand{\theproposition}{S\arabic{proposition}}

\setcounter{lemma}{0}
\renewcommand{\thelemma}{S\arabic{lemma}}

\setcounter{tocdepth}{2} % Show sections and subsections in the TOC.
\tableofcontents

%\section{Related works}

\section{Preliminaries}

\subsection{Interaction-picture evolution}

Here we briefly recall the interaction-picture representation, which will be used throughout this paper. Let
\begin{align}
H=A+\alpha B,\qquad U(\tau):=e^{-i\tau H},
\end{align}
and factor out the evolution generated by $A$ by writing
\begin{align}
U(\tau)=e^{-i\tau A}U_I(\tau),
\end{align}
where the subscript $I$ denotes the interaction-picture evolution. Differentiating this expression gives
\begin{align}
\frac{d}{d\tau}U_I(\tau)&=iAe^{i\tau A}U(\tau)+e^{i\tau A}\frac{dU(\tau)}{d\tau}\nonumber\\
&=iAe^{i\tau A}U(\tau)-ie^{i\tau A}(A+\alpha B)U(\tau)\nonumber\\
&=-i\alpha e^{i\tau A}Be^{-i\tau A}U_I(\tau).
\end{align}
Thus $U_I(\tau)$ satisfies the time-dependent Schr\"odinger equation
\begin{align}
\frac{d}{d\tau}U_I(\tau)=-i\alpha B_I(\tau)U_I(\tau),\qquad B_I(\tau):=e^{i\tau A}Be^{-i\tau A},
\end{align}
with initial condition $U_I(0)=I$. Therefore,
\begin{align}
U_I(\tau)=\mathcal T\exp\left(-i\alpha\int_0^\tau B_I(s)ds\right),
\end{align}
where $\mathcal T$ denotes time ordering. Consequently,
\begin{align}
e^{-i\tau(A+\alpha B)}=e^{-i\tau A}\mathcal T\exp\left(-i\alpha\int_0^\tau B_I(s)ds\right).
\end{align}
Using the Dyson series of the time-ordered exponential in the perturbation parameter $\alpha$~\cite{sakurai2020modern}, we obtain
\begin{align}
e^{-i\tau(A+\alpha B)}=e^{-i\tau A}\left[I-i\alpha\int_0^\tau B_I(s)ds-\alpha^2\int_{0<u<s<\tau}B_I(s)B_I(u)duds+\mathcal O(\alpha^3\tau^3)\right].
\end{align}

\subsection{The mixing lemma}

Our randomized product formulas define mixed-unitary channels after averaging over the random choices. Since the performance guarantees in the main text are stated in diamond distance, we need to compare such averaged channels with the ideal unitary channel. Directly controlling the diamond distance is however often inconvenient. The following sharpened mixing lemma (see also the original mixing lemmas~\cite{campbell2017shorter,hastings2017turning}) allows us to reduce this task to bounding the operator-norm distance between the target unitary and the averaged random unitary. 

\begin{lemma}[Sharpened mixing lemma~\cite{chen2021concentration}]
\label{lem:mixing-lemma}
Let $U$ be a fixed unitary and let $V_\omega$ be a random unitary. Define the ideal unitary channel and the averaged mixed-unitary channel by
\begin{align}
\mathcal U(\rho):=U\rho U^\dagger,\qquad \mathcal V(\rho):=\mathbb E_\omega\left[V_\omega\rho V_\omega^\dagger\right].
\end{align}
Then
\begin{align}
\frac{1}{2}\|\mathcal U-\mathcal V\|_\diamond \le 2\left\|U-\mathbb E_\omega V_\omega\right\|,
\end{align}
where $\|A\|:=\sup_{\|\ket{\psi}\|_2=1}\|A\ket{\psi}\|_2$ denotes the operator (spectral) norm of $A$. 
\end{lemma}

Thus it suffices to control the deterministic bias of the averaged unitary $\mathbb E_\omega V_\omega$ relative to the target unitary $U$. We omit the proof of the lemma and refer to Ref.~\cite{chen2021concentration}. %We use this lemma throughout the SM to prove diamond-distance bounds for the randomized protocols. 

\section{Randomized product formulas in the standard access model}
In this section, we analyze randomized product formulas in the standard access model. We first derive the randomized first-order bound in Eq.~\eqref{eq:randomized-1st-order-bound}, then prove the randomized high-order result in Theorem~\ref{thm:randomized-high-order}, and finally establish the deterministic lower bound in Proposition~\ref{prop:deterministic-barrier-standard}.

\subsection{Derivation of Eq.~\eqref{eq:randomized-1st-order-bound}}

We derive Eq.~\eqref{eq:randomized-1st-order-bound}. Expanding one shifted step through second order in $\tau$ gives
\begin{align}
\mathscr R_1(\tau,u) = I-i\tau H-\frac{\tau^2}{2}(A^2+\alpha^2B^2) -\alpha\tau^2\big[(1-u)AB+uBA\big] + \mathcal O(\tau^3),
\end{align}
whereas
\begin{align}
e^{-i\tau H} = I-i\tau H -\frac{\tau^2}{2}\big[A^2+\alpha(AB+BA) +\alpha^2B^2\big] +\mathcal O(\tau^3).
\end{align}
Consequently,
\begin{align}
\mathscr R_1(\tau,u)-e^{-i\tau H} = \alpha\tau^2\left(u-\frac12\right)[A,B] +\mathcal O(\tau^3).
\end{align}
Since $\mathbb E_u[u-1/2]=0$, the entire second-order term vanishes upon averaging. Moreover, Eqs.~\eqref{eq:exact-evolution-alpha-exapnsion}--\eqref{eq:first-order-randomized-expansion-match} show that the terms linear in $\alpha$ agree exactly for all $\tau$. Expanding one order further therefore gives
\begin{align}
\mathbb E_u\mathscr R_1(\tau,u)-e^{-i\tau H} = -\frac{i\alpha^2\tau^3}{12}[B,[A,B]] +\mathcal O(\alpha^2\tau^4),
\end{align}
and hence, for $\tau = \delta t = t/r$,
\begin{align}
\left\| \mathbb E_u\mathscr R_1(\delta t,u)-e^{-i\delta tH} \right\| = \mathcal O(\alpha^2\delta t^3).
\end{align}

Set $M := \mathbb E_u\mathscr R_1(\delta t,u)$ and $U_\delta:=e^{-i\delta tH}$. By independence of $u_1,\ldots,u_r$,
\begin{align}
\mathbb E_{\mathbf u}\mathscr R_1^{(r)}(t,\mathbf u)=M^r.
\end{align}
Using the telescoping identity
\begin{align}
\label{eq:telescoping-first}
M^r-U_\delta^r
=\sum_{\ell=0}^{r-1}M^{r-1-\ell}(M-U_\delta)U_\delta^\ell,
\end{align}
together with $\|M\| = \|E_u\mathscr R_1(\delta t,u)\| \leq\mathbb E_u\|\mathscr R_1(\delta t,u)\|=1$ and $\|U_\delta\|=1$, Eq.~\eqref{eq:telescoping-first} gives
\begin{align}
\left\|\mathbb E_{\mathbf u}\mathscr R_1^{(r)}(t,\mathbf u)-e^{-itH}\right\| = \|M^r-U_\delta^r\| \le r\|M-U_\delta\| = \mathcal O\left(\frac{\alpha^2t^3}{r^2}\right).
\end{align}
Finally, applying Lemma~\ref{lem:mixing-lemma} to the above equation with $V_{\mathbf u}=\mathscr R_1^{(r)}(t,\mathbf u)$ and $U=e^{-itH}$ yields
\begin{align}
\label{eq}
\left\|\mathcal R_1^{(r)}-\mathcal U_H(t)\right\|_\diamond = \mathcal O\left(\frac{\alpha^2t^3}{r^2}\right),
\end{align}
where $\mathcal U_H(t)(\rho):=e^{-itH}\rho e^{itH}$. This proves Eq.~\eqref{eq:randomized-1st-order-bound}.

\subsection{Proof of Theorem~\ref{thm:randomized-high-order}}
We recall the randomized high-order formula in the standard access model. Let $\delta t=t/r$ and let $u_1,\ldots,u_r\sim\operatorname{Unif}[0,1]$ be independent. The shifted first-order step is
\begin{align}
\mathscr R_1(\tau,u):=e^{-i(1-u)\tau A}e^{-i\alpha\tau B}e^{-iu\tau A},
\end{align}
the second-order step is $\mathscr R_2(\tau,u) :=\mathscr R_1(\tau/2,1-u)\mathscr R_1(\tau/2,u)$, and the higher-order steps are defined recursively by
\begin{align}
\label{eq:high-order-randomized-formula-one-step-sm}
\mathscr R_{2k}(\tau,u) :=\mathscr R_{2k-2}(s_k\tau,u)^2 \mathscr R_{2k-2}((1-4s_k)\tau,u) \mathscr R_{2k-2}(s_k\tau,u)^2,
\end{align}
where $s_k=1/(4-4^{1/(2k-1)})$. The $r$-step random unitary and its averaged channel are
\begin{align}
\mathscr R_{2k}^{(r)}(t,\mathbf u)
&:=\mathscr R_{2k}(\delta t,u_r)\cdots \mathscr R_{2k}(\delta t,u_1), \nonumber \\
\mathcal R_{2k}^{(r)}(\rho)
&:=\mathbb E_{\mathbf u}\left[\mathscr R_{2k}^{(r)}(t,\mathbf u)\rho\mathscr R_{2k}^{(r)}(t,\mathbf u)^\dagger\right].
\end{align}
\begin{theorem}[Restatement of Theorem~\ref{thm:randomized-high-order}]
\label{thm:randomized-high-order-sm}
For each fixed $k\ge1$, the randomized $2k$-th order product formula satisfies
\begin{align}
\left\|\mathcal R_{2k}^{(r)}-\mathcal U_H(t)\right\|_\diamond = \mathcal O\left(\frac{\alpha^2t^{2k+1}}{r^{2k}}\right),
\end{align}
where $\mathcal U_H(t)(\rho):=e^{-itH}\rho e^{itH}$.
\end{theorem}
\begin{proof}
By Lemma~\ref{lem:mixing-lemma} and the same telescoping argument used in the previous subsection, it suffices to establish the one-step bound in the difference between the averaged unitary and the target unitary
\begin{align}
\label{eq:averaged-unitary-difference-thm1-proof-sm}
\left\|\mathbb E_u\mathscr R_{2k}(\delta t,u)-e^{-i\delta t H}\right\| = \mathcal O(\alpha^2\delta t^{2k+1}).
\end{align}
\emph{(i) Scaling in $\delta t$.} We first prove the required dependence on the step size. For each fixed $u$, the shifted second-order formula has local error $\mathcal O(\delta t^3)$. The $\mathcal O(\delta t^{2k+1})$ scaling for $\mathscr R_{2k}$ then follows from Suzuki's recursive construction~\cite{suzuki1990fractal}. We spell out the proof for completeness.

To see the base case of $k=1$, fix $u\in[0,1]$ and recall the second-order formula
\begin{align}
\mathscr R_2(\tau,u) &=\mathscr R_1(\tau/2,1-u) \mathscr R_1(\tau/2,u) \nonumber\\
&=e^{-iu\tau A/2}e^{-i\alpha\tau B/2}
e^{-i(1-u)\tau A}e^{-i\alpha\tau B/2}
e^{-iu\tau A/2}.
\end{align}
This formula is symmetric, i.e., $\mathscr R_2(-\tau,u)=\mathscr R_2(\tau,u)^{-1}$.  It is also first-order consistent, since expanding each exponential to first order gives
\begin{align}
\mathscr R_2(\tau,u) &=I-i\tau\left(\frac{u}{2}A+\frac{\alpha}{2}B+(1-u)A+\frac{\alpha}{2}B+\frac{u}{2}A\right)+\mathcal O(\tau^2) \nonumber\\
&=I-i\tau(A+\alpha B)+\mathcal O(\tau^2).
\end{align}
Let $\Omega_2(\tau,u):=\log \mathscr R_2(\tau,u)$. By symmetry, $\Omega_2(-\tau,u)=\log\mathscr R_2(-\tau,u)  =\log\left(\mathscr R_2(\tau,u)^{-1}\right) = -\log\mathscr R_2(\tau,u)=-\Omega_2(\tau,u)$, so $\Omega_2(\tau,u)$ contains only odd powers of $\tau$. The first-order consistency above therefore implies $\Omega_2(\tau,u)=-i\tau H+\mathcal O(\tau^3)$, which gives the desired result. 

From this point, the higher-order scaling obtained via Suzuki’s recursion~\cite{suzuki1990fractal} is immediate; we nevertheless state it for completeness. Since $\mathscr R_2(\tau,u)$ is symmetric, every recursively defined $\mathscr R_{2j}(\tau,u)$ is also symmetric. Hence its logarithm contains only odd powers of $\tau$. Assume inductively that, for some $k\ge2$,
\begin{align}
\log \mathscr R_{2k-2}(\tau,u)
=-i\tau H+\tau^{2k-1}\Omega_{2k-1}(u)+\mathcal O(\tau^{2k+1}),
\end{align}
where $\Omega_{2k-1}(u)$ is independent of $\tau$. 

Recall the higher-order formula in Eq.~\eqref{eq:high-order-randomized-formula-one-step-sm},
\begin{align}
\mathscr R_{2k}(\tau,u) = \mathscr R_{2k-2}(s_k\tau,u)^2 \mathscr R_{2k-2}((1-4s_k)\tau,u) \mathscr R_{2k-2}(s_k\tau,u)^2.
\end{align}
Using the Baker--Campbell--Hausdorff (BCH) formula and keeping the leading logarithmic error term gives
\begin{align}
\log \mathscr R_{2k}(\tau,u) &=-i(4s_k + (1-4s_k))\tau H+
\left(4s_k^{2k-1}+(1-4s_k)^{2k-1}\right)\tau^{2k-1}\Omega_{2k-1}(u)
+\mathcal O(\tau^{2k+1}) \nonumber\\ &=-i\tau H+\mathcal O(\tau^{2k+1}),
\end{align}
as $4s_k^{2k-1}+(1-4s_k)^{2k-1}=0$. Therefore $\mathscr R_{2k}(\tau,u)=e^{-i\tau H}+\mathcal O(\tau^{2k+1})$. By induction, this holds for every fixed $k\ge1$, for every $u\in[0,1]$.

\emph{(ii) Scaling in $\alpha$.}
We now show that the averaged one-step error has no term linear in $\alpha$. At $\alpha=0$, since all $B$-evolutions become identities and the $A$-times in $\mathscr R_{2k}(\tau,u)$ add to $\tau$, there is no error. 

It remains to show that the coefficient linear in $\alpha$ also vanishes. Let $B_I(s):=e^{isA}Be^{-isA}$. For a shifted first-order substep of signed length $y$,
\begin{align}
\mathscr R_1(y,u) = e^{-iyA}\left(I-i\alpha yB_I(uy)+\mathcal O(\alpha^2y^2)\right).
\end{align}
Averaging over $u\sim\operatorname{Unif}[0,1]$ gives
\begin{align}
\label{eq:sm-BI}
y \mathbb E_u B_I(uy) = \int_0^y B_I(s)ds,
\end{align}
where the integral is understood in the oriented sense if $y<0$. Now consider such a substep after an accumulated $A$-time $x$. Conjugating by this accumulated $A$-evolution shifts the interaction-picture insertion:
\begin{align}
e^{ixA}\left(\int_0^y B_I(s)ds\right)e^{-ixA}
&=\int_0^y e^{ixA}B_I(s)e^{-ixA}ds =\int_0^y B_I(s+x)ds = \int_x^{x+y}B_I(s')ds'.
\end{align}
Thus each averaged shifted first-order substep contributes the interaction-picture integral over its own oriented time interval. To make this explicit, write one realization of $R_{2k}(\tau,u)$ as an ordered product of substeps $R_1(y_j,v_j(u))$, $j=1,\ldots,N$, ordered from right to left, where $y_j$ is the signed length and, for each fixed $j$, $v_j(u)$ is either $u$ or $1-u$, with the choice determined by the position of the substep in the recursion.  Although the shifts $v_j(u)$ are generally correlated, each is marginally uniform on $[0,1]$.  At first order in $\alpha$, each term contains a single $B$-insertion, while all remaining substeps are evaluated at $\alpha=0$ and are therefore independent of $u$.  Consequently, by linearity of expectation, only the marginal distribution of each $v_j(u)$ enters.  Define
\begin{align}
x_0:=0,\qquad x_j:=\sum_{q=1}^j y_q.
\end{align}
By construction of the Suzuki recursion, $\sum_{j=1}^N y_j=\tau$, and hence $x_N=\tau$.  In particular, setting $\alpha=0$ gives $R_{2k}(\tau,u)=e^{-i\tau A}$.  Since each $v_j(u)$ is marginally uniform, Eq.~\eqref{eq:sm-BI}, followed by the shift identity above, shows that the averaged first-order contribution of the $j$th substep is
\begin{align}
\int_{x_{j-1}}^{x_j}B_I(s)ds.
\end{align}
Therefore the total averaged first-order contribution is
\begin{align}
\sum_{j=1}^N\int_{x_{j-1}}^{x_j}B_I(s)ds = \int_{x_0}^{x_N}B_I(s)ds = \int_0^\tau B_I(s)ds,
\end{align}
where the integrals are understood in the oriented sense if some $y_j<0$. This is exactly the coefficient linear in $\alpha$ in the ideal evolution.

\emph{(iii) Joint scaling.} We have so far established the scalings in $\tau$ and $\alpha$ separately. Temporarily writing $R_{2k}(\tau,u;\alpha)$ to make its $\alpha$-dependence explicit, define
\begin{align}
  F_{2k}(\tau,\alpha) :=\mathbb{E}_u R_{2k}(\tau,u;\alpha)
    -e^{-i\tau(A+\alpha B)}.
\end{align}
For each fixed $u\in[0,1]$, $R_{2k}(\tau,u;\alpha)$ is a finite product of matrix exponentials and is therefore jointly entire in $(\tau,\alpha)$.  Moreover, on every compact $K\subset\mathbb{C}^2$, the corresponding matrix-exponential series converge uniformly on $K\times[0,1]$.  The average over $u$ may therefore be taken term by term, so $\mathbb{E}_uR_{2k}(\tau,u;\alpha)$ is jointly entire.  Since the exact evolution is also jointly entire, $F_{2k}$ admits a convergent Taylor expansion near $(0,0)$,
\begin{align}
F_{2k}(\tau,\alpha) = \sum_{p,q\geq0}C_{p,q}\alpha^p\tau^q.
\end{align}
Part (i), applied for every fixed $\alpha$ near zero, implies $C_{p,q}=0$ whenever $q<2k+1$, while part (ii) implies $C_{0,q}=C_{1,q}=0$ for every $q$.  %Hence $F_{2k}(\tau,\alpha) =\alpha^2\tau^{2k+1}G_{2k}(\tau,\alpha)$, where $G_{2k}$ is analytic and locally bounded near $(0,0)$. 
Therefore $F_{2k}(\tau,\alpha) = \mathcal O(\alpha^2\tau^{2k+1})$,

\emph{(iv) Telescoping.} It remains to pass from Eq.~\eqref{eq:averaged-unitary-difference-thm1-proof-sm} to the $r$-step bound. Set $M_k:=\mathbb E_u\mathscr R_{2k}(\delta t,u)$ and $U_\delta:=e^{-i\delta tH}$. By independence of $u_1,\ldots,u_r$, $\mathbb E_{\mathbf u}\mathscr R_{2k}^{(r)}(t,\mathbf u)=M_k^r$. Using the telescoping identity as before,
\begin{align}
M_k^r-U_\delta^r = \sum_{\ell=0}^{r-1}M_k^{r-1-\ell}(M_k-U_\delta)U_\delta^\ell,
\end{align}
together with $\|M_k\|\leq\mathbb E_u\|\mathscr R_{2k}(\delta t,u)\|=1$ and $\|U_\delta\|=1$, gives
\begin{align}
\left\|\mathbb E_{\mathbf u}\mathscr R_{2k}^{(r)}(t,\mathbf u)-e^{-itH}\right\| =\|M_k^r-U_\delta^r\|  \leq r\|M_k-U_\delta\| =\mathcal O\left(\frac{\alpha^2t^{2k+1}}{r^{2k}}\right).
\end{align}
Finally, applying Lemma~\ref{lem:mixing-lemma} to the mixed-unitary channel $\mathcal R_{2k}^{(r)}$ yields the desired bound. 

\end{proof}

\subsection{Proof of Proposition~\ref{prop:deterministic-barrier-standard}}

\begin{proposition}[Restatement of Proposition~\ref{prop:deterministic-barrier-standard}]
Fix $t\neq 0$ and $m\in\mathbb N$. For real coefficients $a_1,\ldots,a_{m+1}$ and $b_1,\ldots,b_m$, define
\begin{align}
\mathscr D_m(t,\alpha) := e^{-ia_{m+1}tA}e^{-ib_m\alpha tB}\cdots e^{-ib_1\alpha tB}e^{-ia_1tA}.
\end{align}
Let $\mathcal D_m(t,\alpha)(\rho) := \mathscr D_m(t,\alpha)\rho \mathscr D_m(t,\alpha)^\dagger$ and $\mathcal U_H(t)(\rho) := e^{-i(A+\alpha B)t}\rho e^{i(A+\alpha B)t}$. 
For every choice of coefficients independent of $A$, $B$, and $\alpha$, there exist Hermitian matrices $A$ and $B$ such that
\begin{align}
\left\| \mathcal D_m(t,\alpha) - \mathcal U_H(t) \right\|_\diamond = \Omega(\alpha).
\end{align}
\end{proposition}

\begin{figure}[t!]
    \centering
    \includegraphics[width=0.5\linewidth]{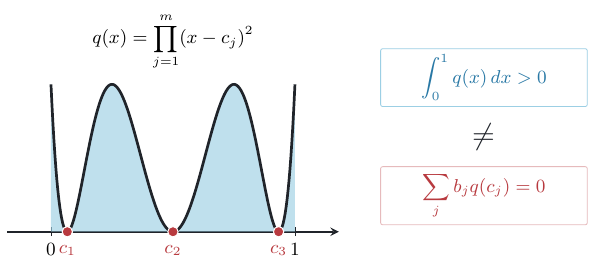}
    \caption{An example of a polynomial witness satisfying $q(c_j)=0$ for all $j$, but $\int_0^1 q(x) dx>0$.}
    \label{fig:polynomial-witness}
\end{figure}

\begin{proof}
Fix arbitrary coefficients and set $c:=\sum_{j=1}^{m+1}a_j$. If $c\neq1$, choose $A:=\operatorname{diag}(0,\pi/((c-1)t))$ and $B:=0$. Then $\mathscr D_m(t,\alpha)=e^{-ictA}$ and $e^{-it(A+\alpha B)}=e^{-itA}$, while $e^{itA}\mathscr D_m(t,\alpha) =e^{-i(c-1)tA}=\operatorname{diag}(1,-1)$. Hence the corresponding unitary channels have diamond distance $2$ for every $\alpha$, which is stronger than the claimed $\Omega(\alpha)$ lower bound. We may therefore assume \begin{align} 
\sum_{j=1}^{m+1}a_j=1. 
\end{align}

Define the partial sums $c_j:=\sum_{\ell=1}^{j}a_\ell$ for $j=1,\ldots,m$. We now compare the terms linear in $\alpha$. Let $B_I(s):=e^{isA}Be^{-isA}$. The exact evolution satisfies
\begin{align}
e^{-it(A+\alpha B)} = e^{-itA}\left(I-i\alpha\int_0^t B_I(s)ds\right) + \mathcal O(\alpha^2).
\end{align}
On the other hand, expanding each weak-$B$ evolution in $\mathscr D_m$ gives
\begin{align}
\mathscr D_m(t,\alpha) = e^{-itA} \left(I-i\alpha t\sum_{j=1}^{m} b_j B_I(tc_j) \right) + \mathcal O(\alpha^2).
\end{align}
Define the first-order mismatch
\begin{align}
\label{eq:deterministic-barrier-first-order-mismatch}
K := t\sum_{j=1}^{m}b_jB_I(tc_j) - \int_0^t B_I(s)ds.
\end{align}
We show that there exist Hermitian $A,B$ for which $K$ has a
nonzero off-diagonal matrix element. More concretely, take
\begin{align}
A=\begin{pmatrix}0&0\\0&\omega\end{pmatrix}, \qquad B=\begin{pmatrix}0&1\\1&0\end{pmatrix},
\end{align}
where $\omega\in\mathbb R$. The $(2,1)$ entry of Eq.~\eqref{eq:deterministic-barrier-first-order-mismatch}, after setting $\lambda=\omega t$ is $K_{21}
=t\bigl(\sum_{j=1}^{m}b_je^{i\lambda c_j}
-\int_0^1e^{i\lambda x}dx\bigr)$.

Suppose, for contradiction, that $K_{21}=0$ for every
$\lambda\in\mathbb R$. Then
\begin{align}
\label{eq:deterministic-barrier-exponential-moment}
\int_0^1 e^{i\lambda x}dx = \sum_{j=1}^{m}b_je^{i\lambda c_j}
\end{align}
for all $\lambda\in\mathbb R$. Differentiating Eq.~\eqref{eq:deterministic-barrier-exponential-moment} at $\lambda=0$ yields
\begin{align}
\int_0^1 x^ndx = \sum_{j=1}^{m} b_j c_j^n, \qquad n=0,1,2,\ldots.
\end{align}
Thus the finite quadrature rule on the right integrates every polynomial exactly on $[0,1]$. Now choose
\begin{align}
q(x):=\prod_{j=1}^{m}(x-c_j)^2.
\end{align}
Then $q(c_j)=0$ for every $j$, so the quadrature rule gives $\sum_{j=1}^m b_j q(c_j)=0$. However, $q(x)\ge0$ on $[0,1]$ and $q$ is not identically zero, so (see Fig.~\ref{fig:polynomial-witness})
\begin{align}
\int_0^1 q(x)dx>0,
\end{align}
a contradiction. 
This contradiction shows that Eq.~\eqref{eq:deterministic-barrier-exponential-moment} cannot hold for every $\lambda\in\mathbb R$. Hence there exists $\lambda_\star\in\mathbb R$ such that $\Delta_\star :=\sum_{j=1}^{m}b_je^{i\lambda_\star c_j} -\int_0^1e^{i\lambda_\star x}dx\neq0$.

Choose $A:=\begin{pmatrix}0&0\\0&\lambda_\star/t\end{pmatrix}$ and $B:=\begin{pmatrix}0&1\\1&0\end{pmatrix}$. For this choice, the first-order mismatch defined in Eq.~\eqref{eq:deterministic-barrier-first-order-mismatch} satisfies $K_{21}=t\Delta_\star\neq0$.

We now derive a lower bound on the diamond distance using a particular input state of $\rho_1:=|1\rangle\langle 1|$. The preceding first-order expansions give $\bigl(\mathcal D_m(t,\alpha)-\mathcal U_H(t)\bigr)(\rho_1) = e^{-itA} \left( -i\alpha[K,\rho_1]+\mathcal O(\alpha^2) \right) e^{itA}$. Since $K$ is Hermitian and $K_{21}\neq0$, we have $-i[K,\rho_1]=\begin{pmatrix}0&iK_{12}\\-iK_{21}&0\end{pmatrix}$, and so that $\left\|-i[K,\rho_1]\right\|_1=2|K_{21}|$.  Therefore,
\begin{align}
\left\| \bigl(\mathcal D_m(t,\alpha)-\mathcal U_H(t)\bigr)(\rho_1) \right\|_1 = 2|\alpha|\,|K_{21}|+\mathcal O(\alpha^2) = \Omega(\alpha),
\end{align}
and, using the definition of the diamond norm, 
\begin{align} 
\left\|\mathcal D_m(t,\alpha)-\mathcal U_H(t)\right\|_\diamond \ge \left\| \bigl(\mathcal D_m(t,\alpha)-\mathcal U_H(t)\bigr)(\rho_1) \right\|_1 = \Omega(\alpha).
\end{align}
\end{proof}

We note that the construction of a polynomial vanishing at the prescribed nodes in the proof above was inspired by a related argument in Ref.~\cite[Lemma~S3]{kim2026high}.

%\subsection{Comparison to standard product formulas}

\section{Randomized THRIFT}

In this section, we analyze randomized product formulas in the stronger access model. We first derive the leading second-order THRIFT error, then construct the randomized correction and prove the base bound in Eq.~\eqref{eq:randomized-thrift-bound}. 
We next prove the randomized high-order result in Theorem~\ref{thm:randomized-thrift}, and finally establish the deterministic lower bound in Proposition~\ref{prop:strong-access-barrier}.

\subsection{Derivation of Eq.~\eqref{eq:thrift-second-order-error}}

Throughout this section, let
\begin{align}
B=\sum_{\ell=1}^L B_\ell,\qquad  B_\ell^{(A)}(s):=e^{isA}B_\ell e^{-isA}.
\end{align}
For an interval $I\subset[0,\tau]$, define the interaction-picture evolution
\begin{align}
W_\ell(I):=\mathcal T\exp\left(-i\alpha\int_I B_\ell^{(A)}(s)ds\right).
\end{align}
In particular, for $I=[x,y]$,
\begin{align}
W_\ell([x,y]) = e^{iyA}e^{-i(y-x)(A+\alpha B_\ell)}e^{-ixA},
\end{align}
so each $W_\ell(I)$ is implementable in the stronger access model. A single THRIFT step in the interaction picture is
\begin{align}
S(\tau) := e^{i\tau A}\mathscr T(\tau) = \prod_{\ell=1}^L W_{\ell}([0,\tau]),
\end{align}
where $\prod_{\ell=1}^L X_\ell:=X_L\cdots X_1$. Expanding the ideal evolution in the interaction picture with respect to $A$,
\begin{align}
e^{-i\tau H} = e^{-i\tau A} U_I(\tau), \qquad  U_I(\tau) := \mathcal T\exp\left(-i\alpha\int_0^\tau\sum_{\ell=1}^L B_\ell^{(A)}(s)ds\right).
\end{align}
We now derive the second-order error of $S(\tau)$. To simplify the notation, define
\begin{align}
X_a:=\int_0^\tau B_{a}^{(A)}(s)ds,\qquad Y_a:=\int_{0<u<s<\tau} B_{a}^{(A)}(s)B_{a}^{(A)}(u)duds .
\end{align}
The Dyson series gives
\begin{align}
W_{a}([0,\tau]) = I-i\alpha X_a-\alpha^2Y_a+\mathcal O(\alpha^3\tau^3).
\end{align}
Hence
\begin{align}
\label{eq:thrift-step-expansion-sm}
S(\tau) = \prod_{\ell=1}^L W_{\ell}([0,\tau]) = I-i\alpha\sum_{a=1}^L X_a -\alpha^2\sum_{a=1}^L Y_a - \alpha^2\sum_{1\le a<b\le L}X_bX_a + \mathcal O(\alpha^3\tau^3).
\end{align}

On the other hand, again by the Dyson expansion, the ideal interaction-picture evolution becomes
\begin{align}
U_I(\tau) &= I -i\alpha\int_0^\tau \sum_{a=1}^L B_{a}^{(A)}(s)ds
-\alpha^2 \int_{0<u<s<\tau} \sum_{a,b=1}^L B_{a}^{(A)}(s)B_{b}^{(A)}(u)duds +\mathcal O(\alpha^3\tau^3) \nonumber\\
&= I -i\alpha\sum_{a=1}^L X_a -\alpha^2\sum_{a=1}^L Y_a -\alpha^2 \sum_{\substack{a,b=1\\a\neq b}}^L \int_{0<u<s<\tau} B_{a}^{(A)}(s)B_{b}^{(A)}(u)duds +\mathcal O(\alpha^3\tau^3) \nonumber\\
&= I -i\alpha\sum_{a=1}^L X_a -\alpha^2\sum_{a=1}^L Y_a -\alpha^2 \sum_{1\le a<b\le L}
\int_{0<u<s<\tau} \left[B_{a}^{(A)}(s)B_{b}^{(A)}(u) + B_{b}^{(A)}(s)B_{a}^{(A)}(u) \right]duds +\mathcal O(\alpha^3\tau^3) \nonumber \\
&= I -i\alpha\sum_{a=1}^L X_a -\alpha^2\sum_{a=1}^L Y_a -\alpha^2\sum_{1\le a<b\le L}X_bX_a
-\alpha^2\sum_{1\le a<b\le L} \int_{0<u<s<\tau} \left[B_{a}^{(A)}(s), B_{b}^{(A)}(u) \right]duds + \mathcal O(\alpha^3\tau^3), \label{eq:ideal-ip-expansion-sm}
\end{align}
where in the last equality we used
\begin{align}
X_bX_a &= \left(\int_0^\tau B_{b}^{(A)}(s)ds\right) \left(\int_0^\tau B_{a}^{(A)}(u)du\right) \nonumber\\
&= \int_0^\tau\int_0^\tau B_{b}^{(A)}(s)B_{a}^{(A)}(u)duds \nonumber\\
&= \int_{0<u<s<\tau} B_{b}^{(A)}(s)B_{a}^{(A)}(u)duds +  \int_{0<s<u<\tau} B_{b}^{(A)}(s)B_{a}^{(A)}(u)duds \nonumber\\
&= \int_{0<u<s<\tau} \left[B_{b}^{(A)}(s)B_{a}^{(A)}(u) + B_{b}^{(A)}(u)B_{a}^{(A)}(s) \right]duds.
\end{align}

Subtracting Eq.~\eqref{eq:ideal-ip-expansion-sm} from Eq.~\eqref{eq:thrift-step-expansion-sm} gives
\begin{align}
\label{eq:thrift-second-order-error-sm}
S(\tau)-U_I(\tau) = \alpha^2 \sum_{1\le a<b\le L} \int_{0<u<s<\tau} \left[B_{a}^{(A)}(s), B_{b}^{(A)}(u) \right]duds + \mathcal O(\alpha^3\tau^3).
\end{align}
This derives Eq.~\eqref{eq:thrift-second-order-error}.

\subsection{Randomized THRIFT construction and derivation of Eq.~\eqref{eq:randomized-thrift-bound}}

\subsubsection{Dyadic correction construction}
\label{sm:subsubsec:dyadic}

Randomized THRIFT appends to $S(\tau)$ a randomly sampled correction unitary built entirely from $W$-blocks, and hence implementable in the same access model. The sampling distribution is chosen so that the averaged second-order contribution of these corrections is the negative of the leading $\alpha^2$ error in Eq.~\eqref{eq:thrift-second-order-error-sm}, thereby canceling it. To define the corrections, for $n\ge0$ and $q=0,\ldots,2^n-1$, let
\begin{align}
Q_{n,q}^{L}:= \left[\frac{q\tau}{2^n},\frac{(2q+1)\tau}{2^{n+1}}\right], \qquad Q_{n,q}^{R}:= \left[\frac{(2q+1)\tau}{2^{n+1}},\frac{(q+1)\tau}{2^n}\right].
\end{align}
Up to measure-zero boundaries, these rectangles decompose the time-ordering triangle as
\begin{align}
\{(u,s):0<u<s<\tau\} = \bigsqcup_{n=0}^{\infty} \bigsqcup_{q=0}^{2^n-1} Q_{n,q}^{L}\times Q_{n,q}^{R}.
\end{align}

For each $1\le a<b\le L$, $n\ge0$, $q=0,\ldots,2^n-1$, and integer $m\ge1$, define the local correction unitary
\begin{align}
\label{eq:local-correction-unitary-sm}
C_{a,b,n,q}^{[m]} := W_{a}(Q_{n,q}^{R})^m W_{b}(Q_{n,q}^{L})^m W_{a}(Q_{n,q}^{R})^{-m} W_{b}(Q_{n,q}^{L})^{-m}.
\end{align}
Each $C_{a,b,n,q}^{[m]}$ is composed only of $W$-blocks and inverse $W$-blocks.

We now define the sampling distribution. Set
\begin{align}
\label{eq:base-construction-beta-sm}
1<\beta<2, \qquad c := \max\left\{1, \left(\frac{2\binom{L}{2}}{1-2^{1-2\beta}}
\right)^{1/2} \right\}, \qquad m_n:=\left\lceil c2^{\beta n}\right\rceil .
\end{align}
For each tuple $(a,b,n,q)$, define
\begin{align}
p_{a,b,n,q}:=\frac{1}{m_n^2}.
\end{align}
The remaining probability is assigned to applying no correction:
\begin{align}
p_0 := 1- \sum_{1\le a<b\le L} \sum_{n=0}^{\infty} \sum_{q=0}^{2^n-1} p_{a,b,n,q}.
\end{align}
With the above choice of $c$, notice that
\begin{align}
1 - p_0 = \sum_{1\le a<b\le L} \sum_{n=0}^{\infty} \sum_{q=0}^{2^n-1} p_{a,b,n,q} &= \sum_{1\le a<b\le L}
\sum_{n=0}^{\infty} 2^n\frac{1}{m_n^2} \nonumber\\
&\le \frac{\binom{L}{2}}{c^2} \sum_{n=0}^{\infty} 2^n2^{-2\beta n} \nonumber\\
&= \frac{\binom{L}{2}}{c^2} \frac{1}{1-2^{1-2\beta}} \le \frac12,
\end{align}
where we used $m_n\ge c2^{\beta n}$ and $\beta>1$. Hence $p_0\in[1/2,1]$; since every $p_{a,b,n,q}\geq0$ and $p_0+\sum_{a,b,n,q}p_{a,b,n,q}=1$ by construction, this defines a valid probability distribution.

\subsubsection{Explicit sampling procedure}
The probabilities defined above specify the weight $p_{a,b,n,q}=m_n^{-2}$ of each correction branch, together with the probability $p_0$ of applying no correction. For an explicit implementation, it is convenient to sample this distribution hierarchically rather than enumerate all tuples $(a,b,n,q)$.

Let
\begin{align}
    N_{\mathrm{pair}}:=\binom{L}{2},
    \qquad
    p_{\mathrm{corr}}
    := \sum_{n=0}^{\infty}
    \frac{N_{\mathrm{pair}}2^n}{m_n^2}
    = 1-p_0.
\end{align}
First, we apply no correction with probability $p_0$. Otherwise, conditioned on applying a correction, we sample the dyadic level $n$ according to
\begin{align}
\Pr(n\mid\mathrm{correction})=\frac{N_{\mathrm{pair}}2^n/m_n^2}{p_{\mathrm{corr}}}.
\end{align}
We then choose $(a,b)$ uniformly from the $N_{\mathrm{pair}}$ pairs satisfying $1\leq a<b\leq L$, and choose $q$ uniformly from $\{0,\ldots,2^n-1\}$. The resulting probability of selecting a particular tuple is
\begin{align}
p_{\mathrm{corr}}\frac{N_{\mathrm{pair}}2^n/m_n^2}{p_{\mathrm{corr}}}\frac{1}{N_{\mathrm{pair}}}\frac{1}{2^n}=\frac{1}{m_n^2}=p_{a,b,n,q},
\end{align}
which reproduces the distribution defined above.

\subsubsection{Error bound on the local correction unitary}

A single randomized THRIFT step in the interaction picture is
\begin{align}
S_{\rm rand}(\tau,\omega) :=
\begin{cases}
S(\tau), & \omega=0,\\[1mm]
C_{a,b,n,q}^{[m_n]}S(\tau), & \omega=(a,b,n,q),
\end{cases}
\end{align}
where $\Pr[\omega=0]=p_0$ and $\Pr[\omega=(a,b,n,q)]=p_{a,b,n,q}$. 

The following lemma shows that each local correction unitary produces, to second order in $\alpha$, the negative of the commutator contribution associated with one dyadic rectangle in the time-ordering triangle.

\begin{lemma}
\label{lem:local-group-commutator}
For every $1\le a<b\le L$, $n\ge0$, $q=0,\ldots,2^n-1$, and $m\ge1$, the local correction unitary in Eq.~\eqref{eq:local-correction-unitary-sm} satisfies
\begin{align}
C_{a,b,n,q}^{[m]} = I-\alpha^2m^2\Gamma_{a,b,n,q} + \mathcal O\left(\alpha^3m^3\ell_n^3\right),
\end{align}
where $\ell_n:=\tau/2^{n+1}$ and 
\begin{align}
\Gamma_{a,b,n,q} := \int_{s\in Q_{n,q}^R} \int_{u\in Q_{n,q}^L} \left[B_{a}^{(A)}(s),
B_{b}^{(A)}(u) \right]duds.
\end{align}
\end{lemma}

\begin{proof}
We write the $\alpha$-dependence explicitly and set $\mathcal A(\alpha):=W_{a}(Q_{n,q}^R;\alpha)^m$ and $\mathcal B(\alpha):=W_{b}(Q_{n,q}^L;\alpha)^m$. Also let $Z_a:=\int_{Q_{n,q}^R}B_{a}^{(A)}(s)ds$ and $Z_b:=\int_{Q_{n,q}^L}B_{b}^{(A)}(u)du$. Then
$\mathcal A(0)=\mathcal B(0)=I$, $\mathcal A'(0)=-imZ_a$, and $\mathcal B'(0)=-imZ_b$. Writing $X:=-imZ_a$ and $Y:=-imZ_b$, we have, to second order,
\begin{align}
\mathcal A(\alpha)&=I+\alpha X+\alpha^2A_2+\mathcal O(\alpha^3), &
\mathcal B(\alpha)&=I+\alpha Y+\alpha^2B_2+\mathcal O(\alpha^3), \nonumber\\
\mathcal A(\alpha)^{-1} &=I-\alpha X+\alpha^2(X^2-A_2)+\mathcal O(\alpha^3), & \mathcal B(\alpha)^{-1} &=I-\alpha Y+\alpha^2(Y^2-B_2)+\mathcal O(\alpha^3).
\end{align}
Hence
\begin{align}
C_{a,b,n,q}^{[m]}(\alpha) &= \mathcal A(\alpha)\mathcal B(\alpha)
\mathcal A(\alpha)^{-1}\mathcal B(\alpha)^{-1} \nonumber\\
&= I+\alpha^2(XY-YX)+\mathcal O(\alpha^3) \nonumber\\
&= I-\alpha^2m^2[Z_a,Z_b]+\mathcal O(\alpha^3) \nonumber\\
&= I-\alpha^2m^2\Gamma_{a,b,n,q}+\mathcal O(\alpha^3).
\end{align}
This is the standard group-commutator identity~\cite{kitaev2002classical}. Equivalently, the expansion above identifies the Taylor coefficients through second order, so Taylor's theorem with integral remainder gives
\begin{align}
C_{a,b,n,q}^{[m]}(\alpha) = I-\alpha^2m^2\Gamma_{a,b,n,q} + \int_0^\alpha
\frac{(\alpha-s)^2}{2} \partial_s^3C_{a,b,n,q}^{[m]}(s)ds .
\end{align}

It remains to bound the remainder uniformly in $m,n,q$. For an interval $I=[x,y]$ of length $h$, the propagator $W_\ell([x,y];\alpha)$ satisfies $\partial_yW_\ell([x,y];\alpha) =-i\alpha B_\ell^{(A)}(y)W_\ell([x,y];\alpha)$. Differentiating this equation with respect to $\alpha$ and iterating gives, for $p=1,2,3$,
\begin{align}
\partial_\alpha^p W_\ell([x,y];\alpha) &=
(-i)^p p! \int_{x\le t_p\le\cdots\le t_1\le y}
W_\ell([t_1,y];\alpha)B_\ell^{(A)}(t_1)
W_\ell([t_2,t_1];\alpha) \cdots B_\ell^{(A)}(t_p)W_\ell([x,t_p];\alpha)
dt_1\cdots dt_p, \nonumber\\
\left\|\partial_\alpha^p W_\ell(I;\alpha)\right\| &\le p! \int_{x\le t_p\le\cdots\le t_1\le y}
\|B_\ell\|^p dt_1\cdots dt_p = \mathcal O(h^p).
\end{align}
where we used unitarity of $W_\ell$. The same estimate holds for $W_\ell(I;\alpha)^{-1}$. Therefore, for $F(\alpha)=W_\ell(I;\alpha)$ or $W_\ell(I;\alpha)^{-1}$,
\begin{align}
\left\|\partial_\alpha^pF(\alpha)^m\right\| &\le \sum_{\substack{r_1+\cdots+r_m=p\\ r_j\ge0}} \frac{p!}{r_1!\cdots r_m!} \prod_{j=1}^m \left\|\partial_\alpha^{r_j}F(\alpha)\right\| = \mathcal O(m^ph^p),
\qquad p=1,2,3.
\end{align}
where the factors with $r_j=0$ are bounded by unitarity. Applying this with $h=\ell_n$ to the four factors $\mathcal A$, $\mathcal B$, $\mathcal A^{-1}$, and $\mathcal B^{-1}$ yields
\begin{align}
\left\| \partial_\alpha^3C_{a,b,n,p}^{[m]}(\alpha) \right\| &\le
\sum_{\substack{p_1+p_2+p_3+p_4=3\\ p_j\ge0}}
\frac{3!}{p_1!p_2!p_3!p_4!}
\prod_{j=1}^4 \mathcal O(m^{p_j}\ell_n^{p_j})
= \mathcal O(m^3\ell_n^3).
\end{align}
Thus
\begin{align}
\left\| \int_0^\alpha \frac{(\alpha-s)^2}{2} \partial_s^3C_{a,b,n,q}^{[m]}(s)ds
\right\| \le \frac{\alpha^3}{6} \max_{s} \left\| \partial_s^3C_{a,b,n,q}^{[m]}(s) \right\| = \mathcal O\left(\alpha^3m^3\ell_n^3\right).
\end{align}
Combining this with the Taylor expansion proves the claim.

\end{proof}

\subsubsection{Derivation of Eq.~\eqref{eq:randomized-thrift-bound}}

In the Schr\"odinger picture, the one-step formula corresponds to
\begin{align}
\mathscr T_{\rm rand}(\tau,\omega) := e^{-i\tau A}S_{\rm rand}(\tau,\omega).
\end{align}
For $\delta t=t/r$ and independent samples
$\boldsymbol\omega=(\omega_1,\ldots,\omega_r)$, define the $r$-step randomized THRIFT formula
\begin{align}
\mathscr T_{\rm rand}^{(r)}(t,\boldsymbol\omega) := \mathscr T_{\rm rand}(\delta t,\omega_r)\cdots \mathscr T_{\rm rand}(\delta t,\omega_1).
\end{align}
Averaging over the samples gives the mixed-unitary channel
\begin{align}
\mathcal T_{\rm rand}^{(r)}(\rho) := \mathbb E_{\boldsymbol\omega} \left[\mathscr T_{\rm rand}^{(r)}(t,\boldsymbol\omega) \rho \mathscr T_{\rm rand}^{(r)}(t,\boldsymbol\omega)^\dagger \right].
\end{align}

We derive the base randomized-THRIFT bound stated in Eq.~\eqref{eq:randomized-thrift-bound}.

\begin{lemma}
\label{lem:base-randomized-thrift}
The randomized THRIFT protocol defined above satisfies
\begin{align}
\left\| \mathcal T_{\rm rand}^{(r)} - \mathcal U_H(t) \right\|_\diamond = \mathcal O\left(\frac{\alpha^3t^3}{r^2}\right),
\end{align}
where $\mathcal U_H(t)(\rho):=e^{-itH}\rho e^{itH}$. Moreover, the average number of gates per Trotter step is $\mathcal O(L)$.
\end{lemma}

\begin{proof}
Since the rectangles $Q_{n,q}^L\times Q_{n,q}^R$ partition the triangle $0<u<s<\tau$, Eq.~\eqref{eq:thrift-second-order-error-sm} gives 
\begin{align} 
S(\tau)-U_I(\tau)=\alpha^2G(\tau)+\mathcal O(\alpha^3\tau^3), \qquad G(\tau) := \sum_{1\le a<b\le L} \sum_{n=0}^{\infty} \sum_{q=0}^{2^n-1} \Gamma_{a,b,n,q}. 
\end{align} 
We now average one randomized step. Since $p_0+\sum_{a,b,n,q}p_{a,b,n,q}=1$, 
\begin{align} 
\mathbb E_\omega S_{\rm rand}(\tau,\omega) &=
p_0S(\tau) + \sum_{1\le a<b\le L} \sum_{n=0}^{\infty} \sum_{q=0}^{2^n-1} p_{a,b,n,q} C_{a,b,n,q}^{[m_n]}S(\tau) \nonumber\\
&= S(\tau) + \sum_{1\le a<b\le L} \sum_{n=0}^{\infty} \sum_{q=0}^{2^n-1} p_{a,b,n,q} \left(C_{a,b,n,q}^{[m_n]}-I\right)S(\tau) \nonumber\\ 
&= S(\tau) -\alpha^2 \sum_{1\le a<b\le L}
\sum_{n=0}^{\infty}
\sum_{q=0}^{2^n-1} \Gamma_{a,b,n,q}S(\tau) + \mathcal O\left(\alpha^3 \sum_{1\le a<b\le L}
\sum_{n=0}^{\infty} \sum_{q=0}^{2^n-1} m_n\ell_n^3 \right) \nonumber\\
&= S(\tau) - \alpha^2G(\tau)S(\tau) + \mathcal O\left( \alpha^3 \sum_{1\le a<b\le L} \sum_{n=0}^{\infty} \sum_{q=0}^{2^n-1} m_n\ell_n^3 \right), 
\end{align} 
where the third equality follows from Lemma~\ref{lem:local-group-commutator} with $m=m_n$ and the choice $p_{a,b,n,q}=m_n^{-2}$, while the last equality uses the definition of $G(\tau)$. The remaining sum is finite because $m_n=\lceil c2^{\beta n}\rceil$, $\ell_n=\tau/2^{n+1}$, and $\beta<2$: 
\begin{align} 
\sum_{1\le a<b\le L} \sum_{n=0}^{\infty} \sum_{q=0}^{2^n-1} m_n\ell_n^3 &= \mathcal O\left( \sum_{n=0}^{\infty} 2^n2^{\beta n}\left(\frac{\tau}{2^{n+1}}\right)^3 \right) \nonumber\\ 
&= \mathcal O\left( \tau^3\sum_{n=0}^{\infty}2^{(\beta-2)n} \right) = \mathcal O(\tau^3). 
\end{align} 
Therefore 
\begin{align} 
\mathbb E_\omega S_{\rm rand}(\tau,\omega)-U_I(\tau) &= S(\tau)-U_I(\tau)-\alpha^2G(\tau)S(\tau) +\mathcal O(\alpha^3\tau^3) \nonumber\\ 
&= \alpha^2G(\tau)\left(I-S(\tau)\right) +\mathcal O(\alpha^3\tau^3) = \mathcal O(\alpha^3\tau^3), 
\end{align} 
where we used $\|G(\tau)\|=\mathcal O(\tau^2)$ and $\|I-S(\tau)\|=\mathcal O(\alpha\tau)$. By unitary invariance of the operator norm, this gives
\begin{align}
\left\|\mathbb E_\omega \mathscr T_{\rm rand}(\tau,\omega) - e^{-i\tau H} \right\| &=
\left\| e^{-i\tau A} \left(\mathbb E_\omega S_{\rm rand}(\tau,\omega)-U_I(\tau) \right) \right\| \nonumber\\
&= \left\| \mathbb E_\omega S_{\rm rand}(\tau,\omega)-U_I(\tau) \right\| = \mathcal O(\alpha^3\tau^3).
\end{align}

Setting $\tau=\delta t=t/r$ and repeating the randomized step independently $r$ times, let
$M:=\mathbb E_\omega\mathscr T_{\rm rand}(\delta t,\omega)$ and $U_\delta:=e^{-i\delta tH}$. Then
$\mathbb E_{\boldsymbol\omega}\mathscr T_{\rm rand}^{(r)}(t,\boldsymbol\omega)=M^r$, and the telescoping identity again gives
\begin{align}
\left\|\mathbb E_{\boldsymbol\omega}\mathscr T_{\rm rand}^{(r)}(t,\boldsymbol\omega) - e^{-itH} \right\| &= \|M^r-U_\delta^r\| \le r\|M-U_\delta\| = \mathcal O\left(\frac{\alpha^3t^3}{r^2}\right).
\end{align}
Therefore, Lemma~\ref{lem:mixing-lemma} yields
\begin{align}
\left\|\mathcal T_{\rm rand}^{(r)}-\mathcal U_H(t)\right\|_\diamond = \mathcal O\left(\frac{\alpha^3t^3}{r^2}\right).
\end{align}

Finally, we bound the expected gate count. The uncorrected THRIFT step uses $L$ $W$-blocks, and a correction indexed by $(a,b,n,q)$ uses $4m_n$ additional $W$-blocks. Hence the expected number of additional $W$-blocks is
\begin{align}
\label{eq:expected-gate-count-rthrift}
\sum_{1\le a<b\le L} \sum_{n=0}^{\infty} \sum_{q=0}^{2^n-1} p_{a,b,n,q} 4m_n &= 4\sum_{1\le a<b\le L} \sum_{n=0}^{\infty}
2^n\frac{1}{m_n} \nonumber\\
&\le \frac{4\binom{L}{2}}{c} \sum_{n=0}^{\infty}2^{(1-\beta)n} & \nonumber \\
& = \mathcal O \left(\frac{L^2}{c}\right)= \mathcal O(L),
\end{align}
where we used $m_n\ge c2^{\beta n}$, $\beta>1$, and, for fixed $\beta\in(1,2)$,
\begin{align}
c = \max\left\{ 1, \left(\frac{2\binom{L}{2}}{1-2^{1-2\beta}} \right)^{1/2} \right\} = \Theta(L).
\end{align} 
Therefore the average number of $W$-blocks per step is $L+\mathcal O(L)=\mathcal O(L)$. Since each $W$-block is implemented using a constant number of elementary exponentials in the stronger access model, this also implies that the expected total gate count per step is $\mathcal O(L)$.

\end{proof}

\paragraph{Tail of the gate count.} The expected gate count analysis guarantee does not imply concentration of the gate count. Let $G_\omega$ denote the number of additional $W$-blocks in one randomized step.  On a correction branch $(a,b,n,q)$, $G_\omega=4m_n$, and this branch is sampled with probability $m_n^{-2}$.  Hence Eq.~\eqref{eq:expected-gate-count-rthrift} gives $\mathbb{E}[G_\omega]<\infty$, while
\begin{align}
\mathbb{E}[G_\omega^2] = \sum_{1\leq a<b\leq L} \sum_{n=0}^{\infty} \sum_{q=0}^{2^n-1} \frac{(4m_n)^2}{m_n^2} = 16\binom{L}{2}\sum_{n=0}^{\infty}2^n = \infty.
\end{align}
Thus the gate count has finite mean but infinite variance.  Truncating the correction level at a finite $N$ makes the gate count bounded and introduces only a tail bias controlled by the omitted probability, as discussed in Sec.~\ref{subsec:gate-cost-evaluation}.

\subsection{Proof of Theorem~\ref{thm:randomized-thrift}}

We first specify the order-dependent sampling distribution that is suppressed in the main text. Fix $k\geq1$ and choose
\begin{align}
1 < \beta_k < 1 + \frac{1}{2k+2}.
\label{eq:high-order-rthrift-beta-sm}
\end{align}
Define 
\begin{align}
c_k := \max\left\{ 1, \left( \frac{2\binom{L}{2}}{1-2^{1-2\beta_k}} \right)^{1/2} \right\}, \qquad 
m_{k,n} := \left\lceil c_k2^{\beta_k n}\right\rceil, \qquad 
p_{a,b,n,q}^{(k)} := \frac{1}{m_{k,n}^2}. 
\label{eq:high-order-rthrift-distribution-sm}
\end{align}
The remaining probability is assigned to the no-correction branch,
\begin{align}
p_0^{(k)} := 1- \sum_{1\leq a<b\leq L} \sum_{n=0}^{\infty} \sum_{q=0}^{2^n-1}
p_{a,b,n,q}^{(k)}.
\end{align}
Note that this is essentially the same as the base construction in Eq.~\eqref{eq:base-construction-beta-sm}, with $\beta_k$ in place of $\beta$ and $c_k$ in place of $c$. The same calculation as in the preceding subsection shows that this is a valid probability distribution. Throughout this subsection, $\mathscr T_{\rm rand}(\tau,\omega)$ denotes the base randomized-THRIFT step constructed using the distribution in Eq.~\eqref{eq:high-order-rthrift-distribution-sm}.

For $\delta t=t/r$ and independent correction labels $\boldsymbol\omega=(\omega_1,\ldots,\omega_r)$, we defined 
\begin{align}
\mathscr T_{{\rm rand},2k}^{(r)}(t,\boldsymbol\omega) &:= \mathscr T_{{\rm rand},2k}(\delta t,\omega_r) \cdots \mathscr T_{{\rm rand},2k}(\delta t,\omega_1),
\nonumber\\
\mathcal T_{{\rm rand},2k}^{(r)}(\rho) &:= \mathbb E_{\boldsymbol\omega} \left[ \mathscr T_{{\rm rand},2k}^{(r)}(t,\boldsymbol\omega) \rho \mathscr T_{{\rm rand},2k}^{(r)}(t,\boldsymbol\omega)^\dagger \right].
\label{eq:high-order-rthrift-channel-sm}
\end{align}

\begin{theorem}[Restatement of Theorem~\ref{thm:randomized-thrift}]
For every fixed $k\geq1$, the randomized $2k$-th order THRIFT formula
satisfies
\begin{align}
\left\| \mathcal T_{{\rm rand},2k}^{(r)} -\mathcal U_H(t) \right\|_\diamond = \mathcal O\left( \frac{\alpha^3t^{2k+1}}{r^{2k}} \right).
\label{eq:high-order-rthrift-global-sm}
\end{align}
Moreover, its expected number of exponentials per step is $\mathcal O(L)$.
\end{theorem}

\begin{proof}
By Lemma~\ref{lem:mixing-lemma} and the same telescoping argument as in the preceding subsection, it suffices to establish the following one-step estimate for the averaged unitary:
\begin{align}
\left\| \mathbb E_{\omega}\!\left[ \mathscr T_{{\rm rand},2k}(\delta t,\omega) \right] -e^{-i\delta t H} \right\| = \mathcal O\left( \alpha^3\delta t^{2k+1} \right).
\label{eq:randomized-high-order-one-step-bound}
\end{align}

\emph{(i) Scaling in $\delta t$.} Fix a correction label $\omega$. Since every $W$-block over an interval of length $\mathcal O(\delta t)$ is $I+\mathcal O(\delta t)$, the group commutator structure of them gives $C_\omega(\delta t)=I+\mathcal O(\delta t^2)$. Hence the randomized base formula remains first-order consistent:
\begin{align}
\mathscr T_{\rm rand}(\delta t,\omega) = I-i\delta t H+\mathcal O(\delta t^2).
\end{align}
Moreover, because the same sample $\omega$ is reused throughout each high-order step, $\mathscr T_{{\rm rand},2}(-\delta t,\omega) = \mathscr T_{{\rm rand},2}(\delta t,\omega)^\dagger$. Thus $\mathscr T_{{\rm rand},2}$ is symmetric and first-order consistent, and hence has local error $\mathcal O(\delta t^3)$. The standard Suzuki-recursion argument~\cite{suzuki1990fractal} (e.g., the one used in the proof of Theorem~\ref{thm:randomized-high-order-sm}) therefore applies here, giving $\mathscr T_{{\rm rand},2k}(\delta t,\omega) = e^{-i\delta t H} + \mathcal O\left(\delta t^{2k+1}\right)$. 

\emph{(ii) Scaling in $\alpha$.} For a (signed) step length $x$, let $U(x;\alpha):=e^{-ix(A+\alpha B)}$. Equation~\eqref{eq:thrift-second-order-error-sm} and Lemma~\ref{lem:local-group-commutator} imply that, for each fixed $\omega$,
\begin{align}
\mathscr T_{\rm rand}(x,\omega;\alpha) = U(x;\alpha)+\alpha^2D_\omega(x)+\mathcal O(\alpha^3), \qquad \mathbb E_\omega[D_\omega(x)]=0.
\label{eq:randomized-base-alpha-expansion}
\end{align}
The mean-zero property follows because $p_{a,b,n,q}^{(k)}m_{k,n}^2=1$ and the dyadic rectangles partition the time-ordering triangle. Taking the adjoint at signed time $-x$ shows that the adjointed base factor $\mathscr T_{\rm rand}(-x,\omega)^\dagger$ has the same expansion with a mean-zero second-order error.

By construction, $\mathscr T_{{\rm rand},2k}(\delta t,\omega;\alpha)$ is a finite product of scaled copies of $\mathscr T_{{\rm rand},2}$. Since each $\mathscr T_{{\rm rand},2}$ consists of one base factor and one adjointed base factor, fully expanding the Suzuki recursion gives
\begin{align}
\mathscr T_{{\rm rand},2k}(\delta t,\omega;\alpha) = \prod_{\nu=1}^{N} \left[ U(x_\nu;\alpha)+\alpha^2D_{\nu,\omega} +\mathcal O_\omega(\alpha^3) \right], \qquad \sum_{\nu=1}^N x_\nu=\delta t.
\end{align}
Since all exact factors are generated by the same Hamiltonian, $\prod_{\nu=1}^{N}U(x_\nu;\alpha)=U(\delta t;\alpha)$. Substituting the expansion above and keeping terms through order $\alpha^2$ gives
\begin{align}
\mathscr T_{{\rm rand},2k}(\delta t,\omega;\alpha)-U(\delta t;\alpha) = \alpha^2\sum_{\nu=1}^N U_{>\nu}(0)D_{\nu,\omega}U_{<\nu}(0) +\mathcal O_\omega(\alpha^3),
\end{align}
where $U_{>\nu}(0):=U(x_N;0)\cdots U(x_{\nu+1};0)$ and $U_{<\nu}(0):=U(x_{\nu-1};0)\cdots U(x_1;0)$. Since $\mathbb E_\omega[D_{\nu,\omega}]=0$ for every $\nu$, each term in the sum vanishes after averaging:
\begin{align}
\mathbb E_\omega \left[ U_{>\nu}(0)D_{\nu,\omega}U_{<\nu}(0) \right] = U_{>\nu}(0)\mathbb E_\omega[D_{\nu,\omega}]U_{<\nu}(0) =0.
\end{align}
Thus the entire contribution at order $\alpha^2$ vanishes after averaging. Reusing the same $\omega$ in every factor does not affect this cancellation, since each term at order $\alpha^2$ contains only one $D_{\nu,\omega}$; terms containing two such error coefficients first appear at order $\alpha^4$. Consequently, for each fixed $\delta t$, $\mathbb E_\omega[\mathscr T_{{\rm rand},2k}(\delta t,\omega;\alpha)] -U(\delta t;\alpha)=\mathcal O(\alpha^3)$. It remains to show that the implied constant in this $\mathcal O(\alpha^3)$ scales as $\mathcal O(\delta t^{2k+1})$, which we do next.

\emph{(iii) Bounding the joint remainder.} 
Let
\begin{align}
F_k(\delta t,\alpha) := \mathbb E_\omega \left[ \mathscr T_{{\rm rand},2k}(\delta t,\omega;\alpha) \right] -U(\delta t;\alpha).
\end{align}
Parts~(i) and~(ii) show, at the level of Taylor coefficients, that the first possible joint contribution to $F_k$ is of order $\alpha^3\delta t^{2k+1}$. Taylor's theorem therefore reduces the desired error bound to controlling $\partial_{\delta t}^{2k+1}\partial_\alpha^{3}F_k$ uniformly near $(\delta t,\alpha)=(0,0)$.

Set $D:=2k+4$ and $h_n:=2^{-(n+1)}$. Denote a level-$n$ interval as $I_{n,q}(\delta t)=[\lambda_{n,q}\delta t,\mu_{n,q}\delta t]$, where $|\mu_{n,q}-\lambda_{n,q}|=\mathcal{O}(h_n)$. After the change of variables $s=\delta t u$, the corresponding $W$-block can be written as
\begin{align}
W_\ell(I_{n,q}(\delta t)) = \mathcal T\exp\left( -i\alpha\delta t \int_{\lambda_{n,q}}^{\mu_{n,q}} B_\ell^{(A)}(\delta t u)du \right).
\end{align}
Since the integration interval has length $\mathcal{O}(h_n)$, differentiating the Dyson series gives, for $1\leq r+s\leq D$,
\begin{align}
\left\| \partial_{\delta t}^{r}\partial_\alpha^{s} W_\ell(I_{n,q}(\delta t)) \right\| \leq K_Dh_n 
\label{eq:single-W-derivative-bound} 
\end{align}
uniformly for $|\delta t| \leq \delta t_0$ and $|\alpha| \leq \alpha_0$, for some fixed $\delta t_0, \alpha_0 > 0$. The same bound holds for the inverse $W$-block.

Now consider $W_\ell(I_{n,q}(\delta t))^{m_{k,n}}$. If a derivative of total order $r+s$ acts on $p$ distinct factors, there are at most $\mathcal O(m_{k,n}^p)$ choices of those factors, while Eq.~\eqref{eq:single-W-derivative-bound} contributes a factor $\mathcal O(h_n^p)$. Hence
\begin{align}
\max_{r+s\leq D} \left\| \partial_{\delta t}^{r}\partial_\alpha^{s} W_\ell(I_{n,q}(\delta t))^{m_{k,n}} \right\| \leq K_D\sum_{p=0}^{D}(m_{k,n}h_n)^p \leq K_D\left[1+(m_{k,n}h_n)^D\right]. \label{eq:repeated-W-derivative-bound}
\end{align}
The same estimate holds for $W_\ell(I_{n,q}(\delta t))^{-m_{k,n}}$.

A correction contains four such repeated $W$-blocks, and $\mathscr T_{{\rm rand},2k}$ contains only a fixed number of corrections for fixed $k$. Applying the product rule once more therefore gives
\begin{align}
\max_{r+s\leq D} \left\| \partial_{\delta t}^{r}\partial_\alpha^{s} \mathscr T_{{\rm rand},2k} (\delta t,(a,b,n,q);\alpha) \right\| \leq K_{k,L,D}\left[1+(m_{k,n}h_n)^D\right]. 
\label{eq:high-order-branch-derivative-bound}
\end{align}
The no-correction branch contains only a fixed number of $W$-blocks and is bounded independently of $n$.

We now average Eq.~\eqref{eq:high-order-branch-derivative-bound} over the correction labels. Using $p_{a,b,n,q}^{(k)}=m_{k,n}^{-2}$ and the fact that there are $2^n$ choices of $q$ at level $n$, we obtain
\begin{align}
\sum_{1\leq a<b\leq L} \sum_{n=0}^{\infty} \sum_{q=0}^{2^n-1} p_{a,b,n,q}^{(k)} \left[1+(m_{k,n}h_n)^D\right] \leq K_{k,L,D} \left[ 1+ \sum_{n=0}^{\infty}2^{(1-2\beta_k)n} + \sum_{n=0}^{\infty} 2^{[1+(D-2)\beta_k-D]n} \right],
\label{eq:averaged-derivative-summability}
\end{align}
where we used $m_{k,n}=\Theta(2^{\beta_k n})$ and $h_n=\Theta(2^{-n})$.
Since $D=2k+4$, the exponent in the last sum is
\begin{align}
1+(D-2)\beta_k-D = (2k+2)\beta_k-(2k+3)<0.
\end{align}
The first series converges because $\beta_k>1$, while the second converges by the assumption $\beta_k<1+1/(2k+2)$. Therefore, Eqs.~\eqref{eq:high-order-branch-derivative-bound} and \eqref{eq:averaged-derivative-summability} imply that the branchwise mixed derivatives through total order $D$ are summable over $\omega$. Dominated convergence theorem then allows differentiation under the expectation through total order $D$, in particular through $\partial_{\delta t}^{2k+1}\partial_\alpha^3$. Parts~(i) and~(ii) therefore imply $\partial_{\delta t}^{r}F_k(0,\alpha)=0$ for $r=0,\ldots,2k$, and $\partial_\alpha^{s}F_k(\delta t,0)=0$ for $s=0,1,2$.
Applying Taylor's theorem with integral remainder successively in $\delta t$ and $\alpha$ gives
\begin{align}
F_k(\delta t,\alpha) = \frac{\alpha^3\delta t^{2k+1}}{2!(2k)!} \int_0^1\int_0^1 (1-u)^{2k}(1-v)^2 \partial_{\delta t}^{2k+1} \partial_\alpha^{3} F_k(u\delta t,v\alpha) dvdu.
\label{eq:joint-taylor-remainder}
\end{align}
By Eqs.~\eqref{eq:high-order-branch-derivative-bound} and
\eqref{eq:averaged-derivative-summability}, there exists a constant
$M_{k,L}<\infty$ such that
\begin{align}
\sup_{|x|\leq\delta_0,~|y|\leq\alpha_0} \left\| \partial_x^{2k+1}\partial_y^3F_k(x,y) \right\|
\leq M_{k,L}.
\label{eq:uniform-mixed-derivative-bound}
\end{align}
Taking the norm of Eq.~\eqref{eq:joint-taylor-remainder} and using Eq.~\eqref{eq:uniform-mixed-derivative-bound} gives
\begin{align}
\|F_k(\delta t,\alpha)\| 
&\leq \frac{M_{k,L}|\alpha|^3|\delta t|^{2k+1}}{2!(2k)!} \int_0^1(1-u)^{2k}du \int_0^1(1-v)^2dv \nonumber\\ 
&= \frac{M_{k,L}}{6(2k+1)!} |\alpha|^3|\delta t|^{2k+1} \nonumber\\ 
&= \mathcal O\left(\alpha^3\delta t^{2k+1} \right),
\end{align}
which shows Eq.~\eqref{eq:randomized-high-order-one-step-bound}.

\emph{(iv) Expected gate count.}
Finally, a $2k$-th order step contains $N_k=2\cdot5^{k-1}$ base randomized-THRIFT factors. Each base factor uses $L$ uncorrected $W$-blocks and, on a level-$n$ correction branch, an additional $4m_{k,n}$ blocks. Therefore,
\begin{align}
\mathbb E_\omega[G_{\rm base}] &= L+ 4\binom{L}{2} \sum_{n=0}^{\infty}\frac{2^n}{m_{k,n}} \nonumber\\ 
&\leq L+ \frac{4\binom{L}{2}}{c_k} \sum_{n=0}^{\infty}2^{(1-\beta_k)n} = \mathcal O(L), 
\end{align} 
where we used $\beta_k>1$ and $c_k=\Omega(L)$. 
Hence $\mathbb E_\omega[G_{2k}] = \mathcal O(5^{k-1}L)$, which is $\mathcal O(L)$ for every fixed $k$.
\end{proof}

\subsection{Proof of Proposition~\ref{prop:strong-access-barrier}}

\begin{proposition}[Restatement of Proposition~\ref{prop:strong-access-barrier}]
Fix $t\ne0$, $L\ge2$, and a finite $M\in\mathbb N$. Set $B_0:=0$. For real coefficients $c_1,\ldots,c_M$ and labels $\ell_1,\ldots,\ell_M\in\{0,1,\ldots,L\}$, define
\begin{align}
\mathscr D_M(t,\alpha):= e^{-ic_Mt(A+\alpha B_{\ell_M})}\cdots
e^{-ic_1t(A+\alpha B_{\ell_1})}.
\end{align}
Let $\mathcal D_M(t, \alpha)(\rho) := \mathscr D_M(t,\alpha) \rho \mathscr D_M(t,\alpha)^\dagger$ and $\mathcal U_H(t)(\rho):=e^{-itH}\rho e^{itH}$ where $H:=A+\alpha\sum_{\ell=1}^L B_\ell$.

For every choice of coefficients and labels independent of $A,B_1,\ldots,B_L$, and $\alpha$, there exist Hermitian matrices $A,B_1,\ldots,B_L$ such that
\begin{align}
\left\|\mathcal D_M(t,\alpha)- \mathcal{U}_H(t)\right\|_\diamond =\Omega(\alpha^2).
\label{eq:strong-access-quantitative-sm}
\end{align}
\end{proposition}

\begin{proof}
It suffices to consider $t>0$. (The case $t<0$ follows by replacing $t$ by $|t|$ and all Hamiltonians by their negatives.)

\emph{(i) Zeroth-order consistency.}
Set
\begin{align}
x_0:=0, \qquad
x_j:=t\sum_{q=1}^{j}c_q,
\quad j=1,\ldots,M.
\label{eq:strong-access-cumulative-times-sm}
\end{align}
At $\alpha=0$,
\begin{align}
\mathscr D_M(t,0)=e^{-ix_MA}.
\end{align}
If $x_M \neq t$, choose $A:=\operatorname{diag}(0,\pi/(x_M-t))$ and set $B_1=\cdots=B_L=0$. Then $e^{itA}e^{-ix_MA}=\operatorname{diag}(1,-1)$, so the corresponding unitary channels have diamond distance $2$ for every $\alpha$, which is stronger than Eq.~\eqref{eq:strong-access-quantitative-sm}. Hence it remains only to consider $x_M=t$. 

\emph{(ii) Interaction-picture factorization.}
We extend the $W$-block notation to any oriented interval $[x,y]\subset \mathbb R$ by
\begin{align}
W_\ell([x,y]) :=e^{iyA}e^{-i(y-x)(A+\alpha B_\ell)}e^{-ixA}.
\label{eq:strong-access-oriented-W-sm}
\end{align}
The same definition applies when $y<x$, and $W_0([x,y])=I$. For $I_j=[x_{j-1},x_j]$, the identity
\begin{align}
e^{-ic_jt(A+\alpha B_{\ell_j})}e^{-ix_{j-1}A}
=e^{-ix_jA}W_{\ell_j}(I_j)
\end{align}
gives, upon iteration,
\begin{align}
\mathscr D_M(t,\alpha) =e^{-itA}S_D(t,\alpha), \qquad
S_D(t,\alpha) :=W_{\ell_M}(I_M)\cdots W_{\ell_1}(I_1).
\label{eq:strong-access-ip-factorization-sm}
\end{align}
For each $j$, define the oriented indicator
\begin{align}
\chi_j(s):=
\begin{cases}
1, & x_{j-1}<s<x_j,\\
-1, & x_j<s<x_{j-1},\\
0, & \text{otherwise}.
\end{cases}
\label{eq:strong-access-oriented-indicator-sm}
\end{align}
The Dyson expansion of the interaction-picture propagator $W_{\ell_j}(I_j)$ gives
\begin{align}
W_{\ell_j}(I_j) = I-i\alpha\int_{\mathbb R} \chi_j(s)B_{\ell_j}^{(A)}(s) ds + \mathcal O(\alpha^2).
\label{eq:strong-access-W-first-order-sm}
\end{align}
For $x_j<x_{j-1}$, the same expansion follows from $W_\ell([x_{j-1},x_j])=W_\ell([x_j,x_{j-1}])^{-1}$, with the sign encoded by $\chi_j$.

Fix two distinct labels $a,b\in\{1,\ldots,L\}$, which is possible because $L\ge2$. We compare the $\mathcal O(\alpha^2)$ terms proportional to the operator word $B_a^{(A)}(s)B_b^{(A)}(u)$. Since each block $W_{\ell_j}$ contains only one perturbation $B_{\ell_j}$, such a mixed-label term cannot arise from the second-order expansion of a single block. It can only arise by taking the first-order term from two different blocks.

Recall that $S_D(t,\alpha)=W_{\ell_M}(I_M)\cdots W_{\ell_1}(I_1)$. Thus, $B_a^{(A)}(s)$ appears to the left of $B_b^{(A)}(u)$ precisely when it is chosen from a block $p$ lying to the left of a block $q$, namely when $p>q$, $\ell_p=a$, and $\ell_q=b$. Multiplying the two first-order terms therefore gives the following contribution:
\begin{align}
-\alpha^2\int_{\mathbb R^2} K_{ab}(s,u) B_a^{(A)}(s)B_b^{(A)}(u)dsdu,
\label{eq:strong-access-circuit-kernel-term-sm}
\end{align}
where
\begin{align}
K_{ab}(s,u) := \sum_{\substack{p>q\\ \ell_p=a,\ell_q=b}} \chi_p(s)\chi_q(u).
\label{eq:strong-access-circuit-kernel-sm}
\end{align}

On the other hand, the exact interaction-picture evolution is
\begin{align}
U_I(t) &:=e^{itA}e^{-it(A+\alpha\sum_{\ell=1}^{L}B_\ell)} =\mathcal T\exp\left(-i\alpha\int_0^t \sum_{\ell=1}^{L}B_\ell^{(A)}(s)ds\right).
\end{align}
In its second-order Dyson expansion, the operator $B_a^{(A)}(s)B_b^{(A)}(u)$ occurs when the time $s$ of the left operator is later than the time $u$ of the right operator. Hence its coefficient is described by the time-ordering kernel
\begin{align}
T_t(s,u):=\mathbf 1_{\{0<u<s<t\}}.
\label{eq:strong-access-ideal-kernel-sm}
\end{align}
Accordingly, the corresponding exact contribution is obtained from Eq.~\eqref{eq:strong-access-circuit-kernel-term-sm} by replacing $K_{ab}$ with $T_t$. (See Fig.~\ref{fig:strong-access-marginal-obstruction-sm}(a,b) for the concrete four-block example $\ell=(b,a,b,a)$.)

\emph{(iii) Kernel mismatch.} We now show that these two kernels cannot agree. Integrating $K_{ab}(s,u)$ over $u$ gives
\begin{align}
k_{ab}(s) := \int_{\mathbb R}K_{ab}(s,u)du = \sum_{\substack{p>q\\ \ell_p=a, \ell_q=b}} (x_q-x_{q-1})\chi_p(s).
\label{eq:strong-access-circuit-marginal-sm}
\end{align}
Because this is a finite linear combination of interval indicators, $k_{ab}$ is a finite step function. By contrast, the corresponding marginal of the ideal kernel is 
\begin{align}
\int_{\mathbb R}T_t(s,u)du=s, \qquad 0<s<t.
\label{eq:strong-access-ideal-marginal-sm}
\end{align}
Since the set of endpoints $\{x_0,\ldots,x_M\}$ is finite, there exists a nonempty open interval $J\subset(0,t)$ containing no endpoint. The function $k_{ab}$ is constant on $J$, whereas $s\mathbf 1_{(0,t)}(s)=s$ varies strictly there, as illustrated for the same four-block example in Fig.~\ref{fig:strong-access-marginal-obstruction-sm}(c). Consequently, 
\begin{align}
g_{ab}(s) := \int_{\mathbb R} \bigl(K_{ab}(s,u)-T_t(s,u)\bigr) du = k_{ab}(s)-s\mathbf 1_{(0,t)}(s)
\label{eq:strong-access-marginal-difference-sm}
\end{align}
does not vanish almost everywhere.

\begin{figure}[t!] 
\centering 
\includegraphics[width=\linewidth]{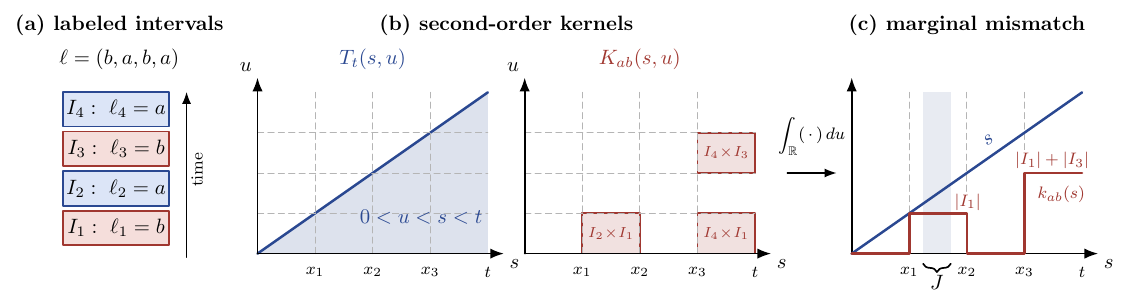} 
\caption{ Kernel and marginal mismatch for the four-block example $\ell=(b,a,b,a)$. (a) The labeled intervals $I_1,\ldots,I_4$ in product order. (b) The ideal time-ordering kernel $T_t(s,u)=\mathbf 1_{\{0<u<s<t\}}$ and the corresponding circuit kernel $K_{ab}(s,u) = \mathbf 1_{I_2\times I_1}(s,u) + \mathbf 1_{I_4\times I_1}(s,u) + \mathbf 1_{I_4\times I_3}(s,u)$. (c) Integrating over $u$ gives the strictly varying ideal marginal $s$ and the circuit marginal $k_{ab}(s)$, which is constant on the shaded open interval $J$. More generally, every finite circuit produces a finite step function $k_{ab}$, so it cannot agree with $s$ throughout $(0,t)$. } 
\label{fig:strong-access-marginal-obstruction-sm} 
\end{figure}

\emph{(iv) Three-dimensional witness.} It remains to turn this kernel mismatch into an operator-norm lower bound. We show that the difference $K_{ab}-T_t$ has a nonzero marginal and then construct a three-dimensional system whose matrix element probes a nonzero Fourier component of this marginal. The construction isolates the ordered term $B_aB_b$, preventing cancellation by the remaining second-order terms. To begin, from Eq.~\eqref{eq:strong-access-marginal-difference-sm} and the uniqueness of the Fourier transform, there exists $\omega\in\mathbb R$ such that
\begin{align}
C_\omega := \int_{\mathbb R}g_{ab}(s)e^{i\omega s} ds \ne0.
\label{eq:strong-access-nonzero-fourier-sm}
\end{align}

On the three-dimensional space with basis $\{\ket{0},\ket{1},\ket{2}\}$, choose
\begin{align}
A =\omega\ket{2} \bra{2}, \qquad
B_a =\ket{2} \bra{1}+\ket{1} \bra{2}, \qquad
B_b=\ket{1} \bra{0}+\ket{0} \bra{1},
\label{eq:strong-access-witness-sm}
\end{align}
and set $B_\ell=0$ for $\ell\notin\{a,b\}$. %These matrices are Hermitian,
For this choice, a single application of either perturbation cannot connect $\ket{0}$ to $\ket{2}$. Hence all first-order contributions vanish in the matrix element between these two states. At second order, the only ordered product that connects $\ket{0}$ to $\ket{2}$ is $B_a^{(A)}(s)B_b^{(A)}(u)$: the right operator first maps $\ket{0}$ to $\ket{1}$, and the left operator then maps $\ket{1}$ to $\ket{2}$. Indeed,
\begin{align}
& B_b^{(A)}(u)\ket{0} =\ket{1}, \qquad
B_a^{(A)}(s)\ket{1} =e^{i\omega s}\ket{2},
\nonumber\\
& \bra{2}B_a^{(A)}(s)B_b^{(A)}(u)\ket{0}
=e^{i\omega s}.
\label{eq:strong-access-surviving-word-sm}
\end{align}
The same-label products $B_aB_a$ and $B_bB_b$, as well as the reverse ordering $B_bB_a$, have zero matrix element between $\ket{0}$ and $\ket{2}$. Therefore, this matrix element isolates precisely the second-order contribution governed by $K_{ab}-T_t$.

Using the definitions of $g_{ab}$ and $C_\omega$, we therefore obtain
\begin{align}
\bra{2}\bigl(S_D(t,\alpha)-U_I(t)\bigr)\ket{0} &= -\alpha^2\int_{\mathbb R^2} \bigl(K_{ab}(s,u)-T_t(s,u)\bigr)
\bra{2}B_a^{(A)}(s)B_b^{(A)}(u)\ket{0}dsdu +\mathcal O(\alpha^3)
\nonumber\\
&= -\alpha^2\int_{\mathbb R^2} \bigl(K_{ab}(s,u)-T_t(s,u)\bigr)e^{i\omega s}dsdu +\mathcal O(\alpha^3)
\nonumber\\
&= -\alpha^2\int_{\mathbb R} g_{ab}(s)e^{i\omega s}ds +\mathcal O(\alpha^3) \nonumber\\
&= -\alpha^2C_\omega+\mathcal O(\alpha^3).
\label{eq:strong-access-witness-error-sm}
\end{align}
%The remainder is $\mathcal O(\alpha^3)$ because $M$ is finite and all matrices are finite-dimensional.
Since $C_\omega\ne0$, there exists $\alpha_0>0$ such that, for every
$0<\alpha<\alpha_0$,
\begin{align}
\left| \bra{2}\bigl(S_D(t,\alpha)-U_I(t)\bigr)\ket{0} \right|
\ge \frac{|C_\omega|}{2}\alpha^2.
\end{align}

%The operator norm is bounded below by any matrix element between normalized vectors. Also, operator norm is unitarily invariant. Hence
%\begin{align} \left\| \mathscr D_M(t,\alpha) - e^{-it(A+\alpha\sum_{\ell=1}^{L}B_\ell)} \right\| &=  \left\|e^{-itA}(S_D(t,\alpha)-U_I(t))\right\| \nonumber\\ &= \left\|S_D(t,\alpha)-U_I(t)\right\| \nonumber\\ &\ge \left| \bra{2}\bigl(S_D(t,\alpha)-U_I(t)\bigr)\ket{0} \right| \nonumber\\ &\ge \frac{|C_\omega|}{2}\alpha^2. \end{align}
%Thus Eq.~\eqref{eq:strong-access-quantitative-sm} holds with $\kappa=|C_\omega|/2$ in the remaining case $x_M=t$. Together with the zeroth-order consistency condition of $x_M\ne t$, this completes the proof.

We now obtain a lower bound on the diamond-norm error by evaluating the two channels on the fixed input state $\rho_0:=|0\rangle\langle0|$. For the witness constructed above, the second-order unitary mismatch appears as a nonzero $(2,0)$ coherence in the difference of the corresponding output states. Define $X_\alpha:=S_D(t,\alpha)\rho_0S_D(t,\alpha)^\dagger-U_I(t)\rho_0U_I(t)^\dagger$. For the witnesses in Eq.~\eqref{eq:strong-access-witness-sm}, neither $B_a$ nor $B_b$ directly connects $|0\rangle$ to $|2\rangle$. Consequently,
\begin{align}
\langle2|S_D(t,\alpha)|0\rangle,
\ \langle2|U_I(t)|0\rangle
&=\mathcal O(\alpha^2), \nonumber \\
\langle0|S_D(t,\alpha)|0\rangle,
\ \langle0|U_I(t)|0\rangle
&=1+\mathcal O(\alpha^2).
\end{align}
It follows from the preceding expansion that $\langle2|X_\alpha|0\rangle = -\alpha^2C_\omega+\mathcal O(\alpha^3)$.  Hence, for all sufficiently small $\alpha>0$,
\begin{align}
|\langle2|X_\alpha|0\rangle|
\ge \frac{|C_\omega|}{2}\alpha^2.
\end{align}
Therefore, we obtain
\begin{align}
\left\|\mathcal D_M(t,\alpha)-\mathcal U_H(t)\right\|_\diamond &\ge \left\|\bigl(\mathcal D_M(t,\alpha)-\mathcal U_H(t)\bigr)(\rho_0) \right\|_1 \nonumber \\
&= \|e^{-itA}X_\alpha e^{itA}\|_1 \nonumber \\
&=\|X_\alpha\|_1 \nonumber \\
&\ge|\langle2|X_\alpha|0\rangle| \nonumber \\
&\ge\frac{|C_\omega|}{2}\alpha^2.
\end{align}

\end{proof}

\section{Numerical simulation details and additional remarks}
\label{subsec:gate-cost-evaluation}

\subsection{Error bounds and gate-cost evaluation}

Here we describe how the error threshold and minimum number of exponentials are evaluated for Figs.~\ref{fig:tfim-numerics} and \ref{fig:randomized_thrift_t9_phase_map}. Throughout, $\delta t=t/r$ denotes one simulation step.

\paragraph{Error bounds and minimum gate count estimates.}
Computing the diamond distance between two quantum channels is computationally too heavy for the system sizes considered in our simulations. We thus use an upper bound on the diamond distance to determine the minimum number of exponentials required to achieve a specified error tolerance.

For a deterministic $r$-step unitary $V_{\rm det}^{(r)}(t)$ and a randomized unitary $V_\omega^{(r)}(t)$, the standard unitary-channel bound (e.g., see Lemma~12 in Ref.~\cite{aharonov1998quantum}) and Lemma~\ref{lem:mixing-lemma} give
\begin{align} 
\epsilon_{\rm det}(r):= \left\| \mathcal{U}_{H}(t)-\mathcal{V}_{\rm det}^{(r)}(t) \right\|_{\diamond} \leq 2\left\|e^{-itH}-V_{\rm det}^{(r)}(t)\right\|, \\
\epsilon_{\rm rnd}(r) := \left\| \mathcal{U}_{H}(t)-\mathcal{V}_{\rm rnd}^{(r)}(t) \right\|_{\diamond} \leq 4\left\|e^{-itH} -\mathbb E_\omega[V_\omega^{(r)}(t)]\right\|,
\end{align}
with both right-hand sides capped at $2$. For a prescribed target precision $\epsilon$, we choose the smallest integer $r$ for which the corresponding upper bound satisfies $\epsilon_{\rm det}(r)\leq\epsilon$ or $\epsilon_{\rm rnd}(r)\leq\epsilon$, respectively.

For the standard-access formulas, independence between steps gives $\mathbb E_{\mathbf u}[\mathscr R_{2k}^{(r)}(t,\mathbf u)] =M_{2k}(\delta t)^r$, where $M_{2k}(\delta t)=\mathbb E_u[\mathscr R_{2k}(\delta t,u)]$. We evaluate this average using a positive Gauss--Legendre quadrature, $M_{2k,Q}(\delta t)=\sum_{j=1}^{Q}w_j\mathscr R_{2k}(\delta t,u_j)$, and increase $Q$ until the reported bounds and minimizing values of $r$ are unchanged to the displayed precision.

For randomized THRIFT, define $M_\infty(\delta t) =\mathbb E_\omega[\mathscr T_{\rm rand}(\delta t,\omega)]$. We explicitly include correction levels $n\leq N$ in $M_N(\delta t)$ and assign the omitted probability to the uncorrected THRIFT branch. The omitted tail satisfies $\|M_\infty(\delta t)-M_N(\delta t)\|\leq2q_N$, where
\begin{align}
q_N=\sum_{n>N}\frac{2^n}{m_n^2} \leq \bar q_N:= \frac{2^{(1-2\beta)(N+1)}}{c^2(1-2^{1-2\beta})}.
\end{align}
Hence, the ideal-ensemble error is certified by
\begin{align}
\epsilon_{\rm rnd}(r) \leq 4\left\|e^{-itH}-M_N(\delta t)^r\right\| +8r\bar q_N.
\end{align}
We use $N=12$, for which $\bar q_{12}=2.03\times10^{-10}$, and report the smallest integer $r$ satisfying the corresponding target error.

\paragraph{Expected randomized-THRIFT gate counts.}
We use the infinite correction ensemble in Eq.~\eqref{eq:local-correction-unitary-sm} with $L=2$ and $\beta=1.7$. Thus $c=\sqrt{\frac{2}{1-2^{-2.4}}}$ and $m_n=\left\lceil c\,2^{1.7n}\right\rceil$. Each branch $(n,q)$, with $q=0,\ldots,2^n-1$, has probability $p_{n,q}=m_n^{-2}$, while the remaining probability $p_{\varnothing}=1-\sum_{n=0}^{\infty}2^n/m_n^2$ is assigned to the uncorrected branch. After merging adjacent exponentials, the expected number of accessible exponentials per randomized step is
\begin{align}
3+\sum_{n=0}^{\infty} \left( \frac{8\times 2^n}{m_n}-\frac{1}{m_n^2} \right) =14.350515\ldots.
\end{align}
Consequently, for this example with $L=2$, deterministic first-order, deterministic second-order, and randomized first-order THRIFT use $3r$, $4r+1$, and approximately $14.35r$ exponentials, respectively.

\begin{figure}[t!]
    \centering
    \includegraphics[width=1\linewidth]{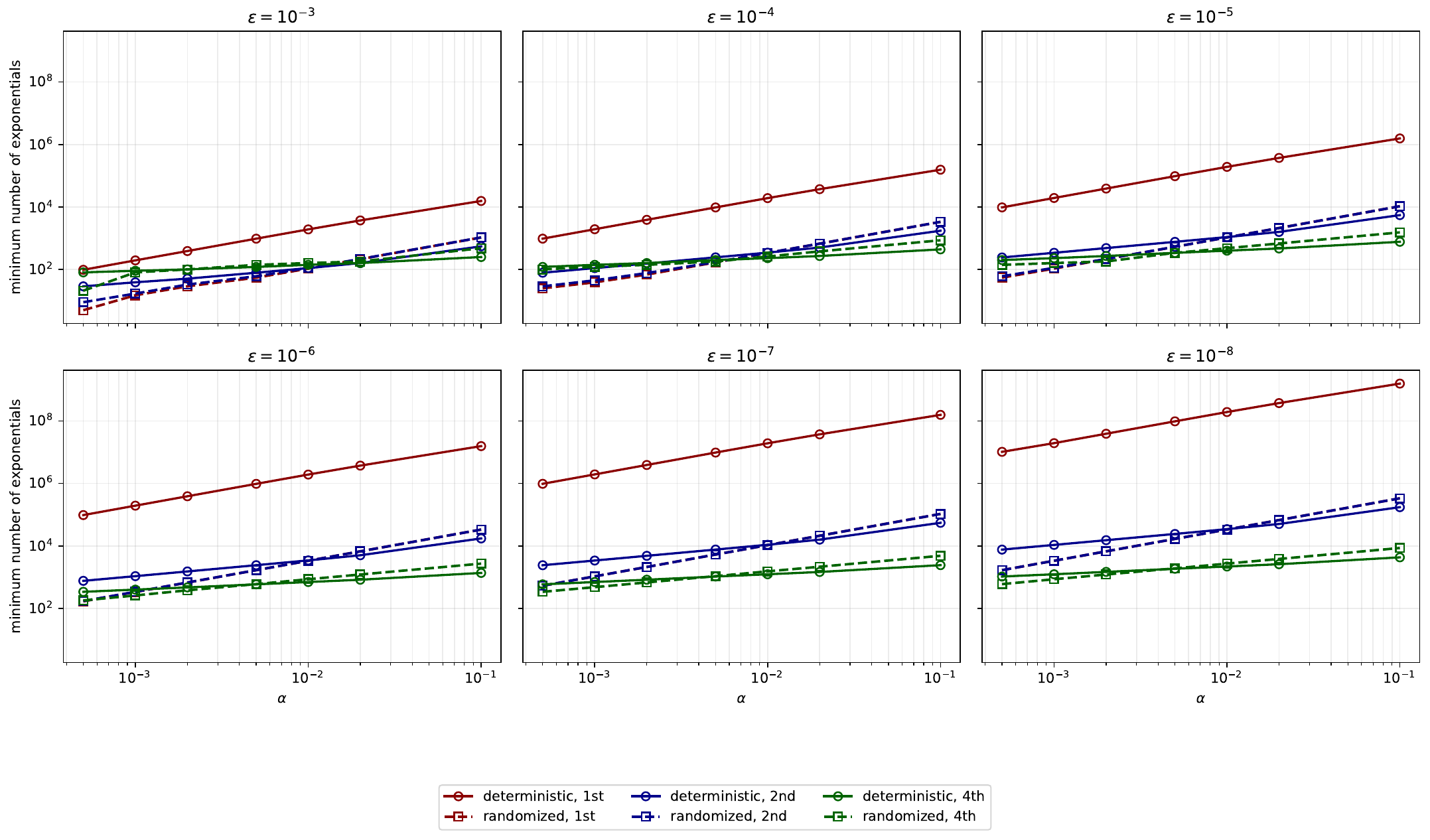}
    \caption{Minimum number of elementary exponentials required for the 9-qubit TFIM at $t=9$. Each panel corresponds to a different target error $\epsilon$ and compares deterministic and randomized formulas of first, second, and fourth order.}
    \label{fig:tfim-precision-gate-counts}
\end{figure}

\subsection{Gate-count crossovers across target precisions}

Figure~\ref{fig:tfim-numerics} in the main text reports the gate-count comparison at the fixed target precision $\epsilon=10^{-4}$. To examine the precision dependence, we repeat the same 9-qubit TFIM benchmark at $t=9$ for $\epsilon\in\{10^{-3},10^{-4},\ldots,10^{-8}\}$. For each pair $(\alpha,\epsilon)$, we minimize the number of exponentials/gates over the integer step number $r$, subject to the certified error bound described in the previous subsection.

\begin{figure}[t!]
    \centering
    \includegraphics[width=0.8\linewidth]{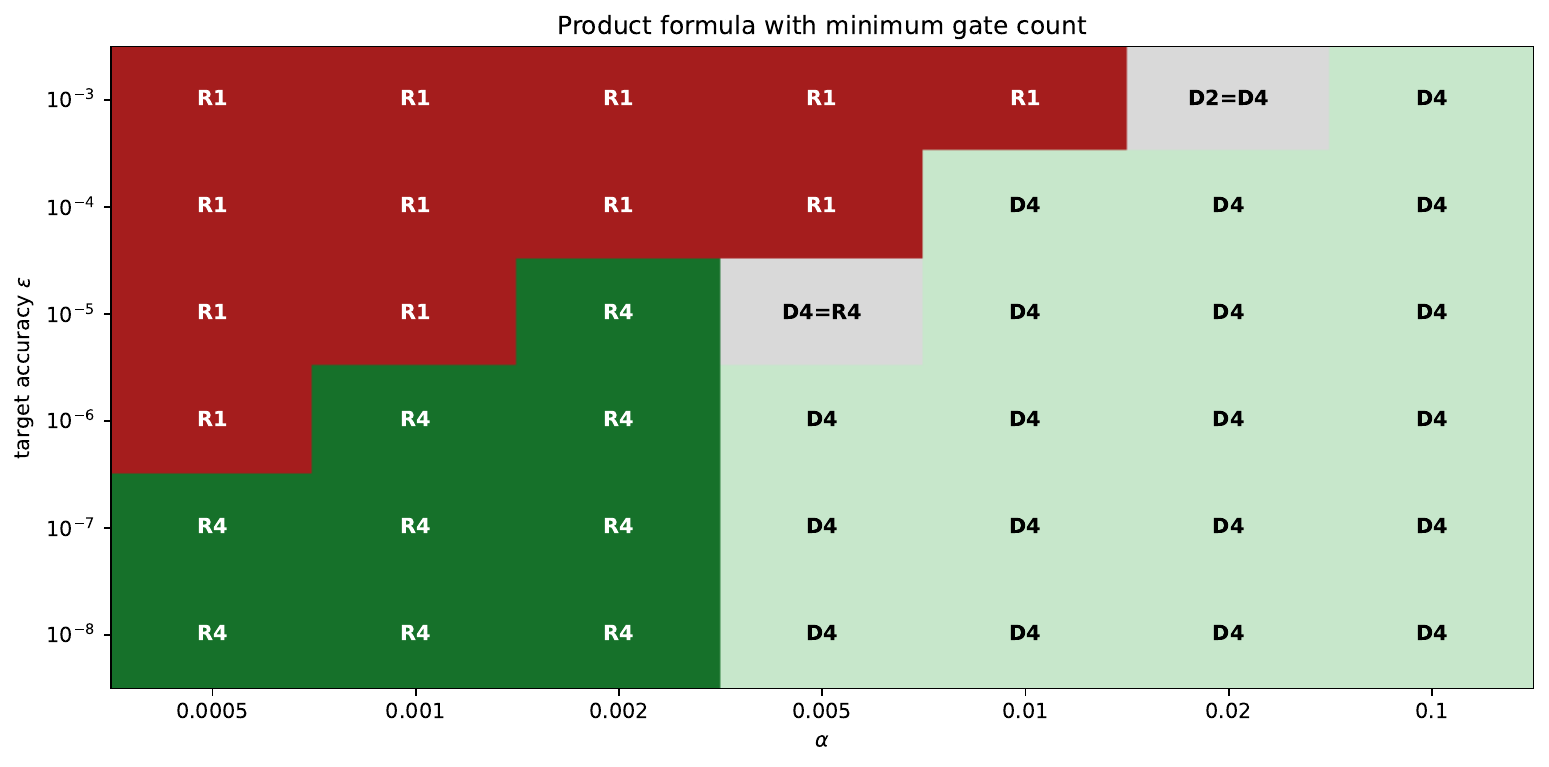}
    \caption{Gate-optimal product formula at each scanned pair $(\alpha,\epsilon)$. Here, $\mathrm{D}_k$ and $\mathrm{R}_k$ denote deterministic and randomized formulas of order $k$ (up to fourth-order here), respectively. An equality sign denotes a tie in the compiled exponential count.}
    \label{fig:tfim-gate-optimal-map}
\end{figure}

The detailed gate counts are shown in Fig.~\ref{fig:tfim-precision-gate-counts} and summarized by the gate-optimal formula map in Fig.~\ref{fig:tfim-gate-optimal-map}. A simple scaling argument explains the observed crossovers. At fixed $t$ and up to method-dependent prefactors,
\begin{align}
    N_{\mathrm{R}1},N_{\mathrm{R}2}
    \propto \left(\frac{\alpha^2}{\epsilon}\right)^{1/2},
    \qquad
    N_{\mathrm{R}4}
    \propto \left(\frac{\alpha^2}{\epsilon}\right)^{1/4},
    \qquad
    N_{\mathrm{D}4}
    \propto \left(\frac{\alpha}{\epsilon}\right)^{1/4}.
\end{align}
where $N_{\mathrm{X}}$ denotes the minimum number of exponentials required by formula $\mathrm{X}$ to achieve the target error $\epsilon$. Tightening the target precision therefore favors higher-order over the lower-order randomized formulas because of its weaker dependence on $\epsilon$. However, $\mathrm{R}_4$ also has a larger per-step cost, so its improved dependence on $\alpha$ is beneficial only when $\alpha$ is sufficiently small. This agrees with the figures: $\mathrm{R}_4$ becomes optimal in the high-precision, small-$\alpha$ regime, while $\mathrm{D}_4$ dominates at larger $\alpha$ and $\mathrm{R}_1$ remains preferable at lower precision. Although we considered only formulas up to fourth order, we therefore expect higher-order formulas to become advantageous at sufficiently high target precision.

\subsection{Finite-sampling error}

The error bounds above assume exact expectation values and, for randomized formulas, exact ensemble averages. In practice, let $X_s\in[-1,1]$ denote the outcome of the $s$-th measurement of an observable $O$ with $\|O\|\leq1$, and define
\begin{align}
\widehat{\mu}_S
:=\frac{1}{S}\sum_{s=1}^{S}X_s,
\qquad
\mu_{\rm alg}
:=\operatorname{Tr}\left[O\mathcal V^{(r)}(\rho)\right],
\qquad
\mu_{\rm ideal}
:=\operatorname{Tr}\left[O\mathcal U_H(t)(\rho)\right].
\end{align}
For randomized formulas, the circuit realization is sampled independently in each execution. Since $\mathbb E[\widehat{\mu}_S]=\mu_{\rm alg}$, Hoeffding's inequality implies that, with probability at least $1-\delta$,
\begin{align}
\left|\widehat{\mu}_S-\mu_{\rm alg}\right|
\leq
\sqrt{\frac{2\ln(2/\delta)}{S}}.
\end{align}
On the other hand, the algorithmic error satisfies
\begin{align}
\left|\mu_{\rm alg}-\mu_{\rm ideal}\right|
&\leq
\|O\|
\left\|
\bigl(\mathcal V^{(r)}-\mathcal U_H(t)\bigr)(\rho)
\right\|_1 \leq
\left\|\mathcal V^{(r)}-\mathcal U_H(t)\right\|_\diamond
= \epsilon_{\rm alg}(r),
\end{align}
where $\epsilon_{\rm alg}(r)$ denotes the corresponding deterministic or randomized error bound derived in the main text. Therefore, by the triangle inequality, with probability at least $1-\delta$,
\begin{align}
\left|\widehat{\mu}_S-\mu_{\rm ideal}\right|
\leq
\epsilon_{\rm alg}(r)
+{\sqrt{\frac{2\ln(2/\delta)}{S}}}.
\end{align}
Thus, finite sampling contributes a universal additive $\mathcal O(S^{-1/2})$ error to each algorithmic error bound.

\end{document}